\documentclass[journal]{IEEEtran}

\usepackage{setspace,slashbox,multirow,tikz}
\usetikzlibrary{arrows,automata}

\usepackage[T1]{fontenc}
\usepackage[latin9]{inputenc}
\usepackage{amsthm}
\usepackage{amsmath}
\usepackage{graphicx}
\usepackage{amssymb}
\usepackage{url}
\usepackage{verbatim}

\theoremstyle{plain}
\newtheorem{thm}{Theorem}
\theoremstyle{plain}

\usepackage{cite,amsfonts,times,bm,amsthm}
\theoremstyle{definition}

\newtheorem{definition}{Definition}
\newtheorem{algo}{Algorithm}

\newtheorem{lemma}{Lemma}
\newtheorem{example}{Example}
\newtheorem*{lemma*}{Lemma}
\newtheorem{corollary}{Corollary}
\theoremstyle{remark}
\newtheorem{remark}{Remark}

\newcommand\T{\rule{0pt}{3.0ex}}
\newcommand\B{\rule[-0.8ex]{0pt}{0pt}}

\newcommand\K{\mathrm{K}}
\newcommand\U{\mathrm{U}}

\newcommand\I{\mathrm{I}}
\newcommand\A{\mathrm{A}}

\title{{\huge Cognitive Access Policies under a Primary ARQ process
 via Forward-Backward Interference Cancellation}}

\author{
\textbf{
Nicol\`o Michelusi, 
Petar Popovski, 
Osvaldo Simeone,\\ 
Marco Levorato, 
Michele Zorzi 
}
\thanks{
N. Michelusi and M. Zorzi are with the 
Department of Information Engineering,
University of Padova, Italy (\texttt{\{michelusi,zorzi\}@dei.unipd.it});
P. Popovski is with the Department of Electronic Systems,
Aalborg University, Denmark (\texttt{petarp@es.aau.dk});
O. Simeone is with the 
Center for Wireless Communication and Signal Processing Research (CWCSPR), New Jersey Institute of Technology,
 New Jersey, USA (\texttt{osvaldo.simeone@njit.edu});
M. Levorato is with the Ming Hsieh Department of Electrical Engineering, University of Southern California,
Los Angeles, USA,
and with Stanford University, USA (\texttt{levorato@stanford.edu}).

This work is a generalization of~\cite{MichelusiITA11},
 presented at the \emph{Information Theory and Applications Workshop} in February 2011, and 
of~\cite{MichelusiBIC11}, presented at the
 \emph{49th Annual Allerton Conference on Communication, Control, and Computing} in September 2011.

 Manuscript received date: Apr. 15, 2012.

 Manuscript revised date: Sep. 1, 2012.
}}

\begin{document}
\maketitle
\begin{abstract}
This paper introduces a novel technique for access by a cognitive Secondary User (SU)
using best-effort transmission
to a spectrum with an incumbent Primary User (PU), which uses Type-I
Hybrid ARQ. The technique leverages the primary ARQ protocol to perform Interference Cancellation
(IC) at the SU receiver (SUrx). Two IC mechanisms that work in concert are introduced: \emph{Forward
IC}, where SUrx, after decoding the PU message, cancels its interference in the (possible) following
PU retransmissions of the same message, to improve the SU throughput; \emph{Backward IC}, where SUrx performs IC on
previous SU transmissions, whose decoding failed due to severe PU interference. Secondary access policies are
designed that determine the secondary access probability in each state of the network so as to maximize
the average long-term SU throughput by opportunistically leveraging IC, while causing bounded average
long-term PU throughput degradation and SU power expenditure. It is proved that the optimal policy
prescribes that the SU prioritizes its access in the states where SUrx knows the PU message, thus enabling IC.
 An algorithm is provided to optimally allocate additional secondary access
opportunities in the states where the PU message is unknown. Numerical results are shown to assess
the throughput gain provided by the proposed techniques.
\end{abstract}
\begin{IEEEkeywords}
Cognitive radios, resource allocation, Markov decision
processes, ARQ, interference cancellation
\end{IEEEkeywords}
\section{Introduction}
Cognitive Radios (CRs)~\cite{Mitola} offer a novel paradigm for improving the efficiency of spectrum
usage in wireless networks. Smart users, referred to as \emph{Secondary Users} (SUs),
adapt their operation in order to opportunistically leverage the channel resource while generating
bounded interference to the Primary Users (PUs)~\cite{FCC,spectrumsharing,Peha}.
For a survey on cognitive radio, dynamic spectrum access and the related research challenges,
we refer the interested reader to~\cite{Peha,cognitiveStart2,goldsmith,DySpAN}.

In a standard model for cognitive radio,
the PU is a legacy system oblivious to the presence of the SU, which needs to satisfy 
given constraints
on the performance loss caused to the PU
(underlay cognitive radio paradigm~\cite{goldsmith}). 
Within this framework, we propose to exploit the intrinsic redundancy,
in the form of copies of PU packets, introduced by the 
Type-I Hybrid Automatic Retransmission reQuest (Type-I HARQ~\cite{Comroe}) protocol
 implemented by the PU by enabling Interference Cancellation (IC)
at the SU receiver (SUrx). 
We introduce two IC schemes that work in concert, both
 enabled by the underlying retransmission process of the PU.
With \emph{Forward IC} (FIC),
 SUrx, after decoding the PU message,
performs IC in 
the next PU retransmission attempts, if these occur. While FIC provides IC on 
SU transmissions performed in future
time-slots, \emph{Backward IC} (BIC) provides IC on 
SU transmissions performed in previous time-slots within the same  primary ARQ
retransmission window, whose decoding failed due to severe interference from the PU.
BIC relies on buffering of the received signals.
Based on these IC schemes, we model the state evolution of the PU-SU network 
as a Markov Decision Process~\cite{Bertsekas,DJWhite},
induced by the specific access policy used by the SU,
which determines its access probability in each state
 of the network.
Following the approach put forth by~\cite{IT_ARQ},
we study the problem of designing optimal secondary access policies 
that maximize the average long-term SU throughput by opportunistically 
leveraging FIC and BIC,
while causing 
a bounded average long-term throughput loss to the PU and a bounded average long-term SU power expenditure.
We show that the optimal strategy dictates that the SU prioritizes its channel access in the states
where SUrx knows the PU message, thus enabling IC; moreover,
we provide an algorithm to optimally allocate
additional secondary access opportunities in the states 
where the PU message is unknown.

The idea of exploiting PU retransmissions to perform IC on future packets (similar to our FIC mechanism) was 
put forth by~\cite{Nosratinia},
which devises several cognitive radio protocols exploiting the hybrid ARQ retransmissions of the PU.
Therein, the PU employs hybrid ARQ with incremental redundancy and the ARQ mechanism is limited to at most one retransmission.
The SU receiver attempts to decode the PU message in the first time-slot. If successful, the SU transmitter
sends its packet and the SU receiver decodes it by using IC on the received signal.
In contrast, in this work, we address the more general case of an arbitrary number of
primary ARQ retransmissions, and we allow a more general access pattern for the SU pair over the entire 
primary ARQ window. We also model the interplay between
the primary ARQ protocol and the activity of the SU, by allowing for BIC.
 It should be noted that IC-related schemes are also used in other context, \emph{e.g.}, decoding for graphical codes
~\cite{LT} and multiple access protocols~\cite{Katabi}.

Other related works include~\cite{Nosratinia2}, which devises an opportunistic sharing scheme with channel probing
 based on the ARQ feedback from the PU receiver.
An information theoretic framework for cognitive radio is investigated in~\cite{Jovicic}, where the SU transmitter has non-casual 
knowledge of the PU's codeword. In~\cite{Devroye}, the data transmitted by the PU is obtained causally at the SU receiver.
However, this model requires a joint design of the PU and SU signaling and channel state information at the transmitters.
In contrast, in our work we explicitly model the dynamic acquisition of the PU message at the SU receiver, which enables IC.
Moreover, the PU is oblivious to the presence of the SU.

The paper is organized as follows. 
Sec.~\ref{sec:sys_model} presents the system model.
Sec.~\ref{sec:performance} introduces the
secondary access policy, the
 performance metrics and the optimization problem,
which is addressed in
Sec.~\ref{sec:SU_tx_policy}. 
Sec.~\ref{sec:numres} presents and discusses the numerical results.
Finally, Sec.~\ref{sec:remarks} concludes the paper. 
The proofs of the lemmas and theorems are provided in the appendix.
\section{System Model}
\label{sec:sys_model}
We consider a two-user interference network, as depicted in Fig.~\ref{fig:cog_net},
 where a primary transmitter and a secondary transmitter, denoted by PUtx 
 and SUtx, respectively,
transmit to their respective receivers, PUrx and SUrx, over 
the direct links PUtx$\rightarrow$PUrx and SUtx$\rightarrow$SUrx.
 Their transmissions generate mutual interference over the links
PUtx$\rightarrow$SUrx and SUtx$\rightarrow$PUrx.

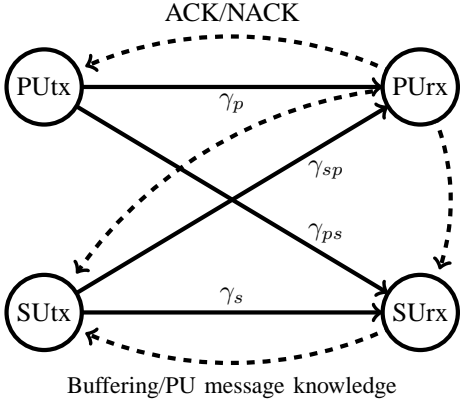
\begin{figure}
    \centering
    {
\begin{tikzpicture}
\draw [ultra thick, ->] (0,0) -- (5-0.5,0);
\draw [ultra thick, ->] (0,3) -- (5-0.5,3);
\draw [ultra thick, ->] (0,0) -- (5-5*0.0857,3-3*0.0857);
\draw [ultra thick, ->] (0,3) -- (5-5*0.0857,0+3*0.0857);
\draw [dashed, ultra thick, ->] (5,3) arc [radius=5, start angle=60, end angle= 113];
\draw [dashed, ultra thick, ->] (5,3) arc [radius=5.83, start angle=92, end angle= 145];
\draw [dashed, ultra thick, ->] (5,0) arc [radius=5, start angle=-60, end angle=-113];
\draw [dashed, ultra thick, ->] (5,3) arc [radius=3, start angle=30, end angle=-18];
\draw [fill=white, ultra thick] (0,0) circle [radius=0.5];;
\draw [fill=white, ultra thick] (5,0) circle [radius=0.5];;
\draw [fill=white, ultra thick] (0,3) circle [radius=0.5];;
\draw [fill=white, ultra thick] (5,3) circle [radius=0.5];;
\node at (0,0) {SUtx};
\node at (5,0) {SUrx};
\node at (0,3) {PUtx};
\node at (5,3) {PUrx};
\node at (2.5,+0.22) {$\gamma_s$};
\node at (2.5,3-0.22) {$\gamma_p$};
\node at (5-1.25,3-1.1) {$\gamma_{sp}$};
\node at (5-1.25,1.05) {$\gamma_{ps}$};
\node at (2.5,3+1) {ACK/NACK};
\node at (2.5,-1) {\small Buffering/PU message knowledge};
\end{tikzpicture}}
\caption{System model}
\label{fig:cog_net}
\end{figure}

Time is divided into time-slots of fixed duration. Each time-slot matches the length of 
the PU and SU packets,
and the transmissions of the PU and SU are assumed to be perfectly synchronized.
We adopt the block-fading channel model, \emph{i.e.}, the channel gains are
constant within the time-slot duration, and change from time-slot to time-slot.
Assuming that the SU and the PU transmit with constant power $P_s$ and $P_p$,
 respectively, and that noise at the receivers is zero mean Gaussian with variance $\sigma_w^2$,
we define the  instantaneous Signal to Noise Ratios (SNR)
of the links SUtx$\rightarrow$SUrx, PUtx$\rightarrow$PUrx, SUtx$\rightarrow$PUrx and PUtx$\rightarrow$SUrx,
during the $n$th time-slot,
 as $\gamma_s(n)$, $\gamma_p(n)$, $\gamma_{sp}(n)$ and $\gamma_{ps}(n)$, respectively.
We model the SNR process $\left\{\gamma_x(n),n=0,1,\dots\right\}$,
where $x\in\{s,p,sp,ps\}$, as i.i.d. over time-slots and independent over the different links,
and we denote the average SNR as $\bar\gamma_x=\mathbb E[\gamma_x]$.

We assume that no Channel State Information (CSI) is available at the transmitters, so that
the latter cannot allocate their rate
based on the instantaneous link quality, to ensure correct delivery
of the packets to their respective receivers. Transmissions
may thus undergo outage, when the selected rate is not supported by the current
channel quality.

In order to improve reliability, the PU employs Type-I HARQ~\cite{Comroe}
with deadline $D\geq 1$,
\emph{i.e.}, at most $D$ transmissions of the same PU message can be performed,
after which the packet is discarded
and a new transmission is performed (the PU is assumed to be backlogged).
We define the
\emph{primary ARQ state} $t\in\mathbb{N}(1,D)$\footnote{We define $\mathbb{N}(n_{0},n_{1})=\left\{ t\in\mathbb{N},n_{0}\leq t\leq n_{1}\right\}$
for $n_{0}\leq n_{1}\in\mathbb{N}$} as the number of ARQ
transmission attempts already performed on the current PU message, plus the current one. Namely, $t=1$ indicates a new PU
transmission, and the counter $t$ is increased at each ARQ retransmission,
until the deadline $D$ is reached.
We assume that the ARQ feedback is received at the PU transmitter by the end of the time-slot, so that, if requested, a retransmission
can be performed in the next time-slot.

On the other hand, the SU, in each time-slot,
either accesses the channel by transmitting its own message, or stays idle.
This decision is based on the access policy $\mu$, defined in Sec.~\ref{sec:performance}.
The activity of the SU, which is governed by $\mu$, 
affects the outage performance of the PU, by creating interference to the PU over the link SUtx$\rightarrow$PUrx.
We denote the primary outage probability when the SU is idle and accesses the channel, respectively, as\footnote{
Herein, we denote the outage probability as $q_{xy}^{(Z)}$,
where $x$ and $y$ are the source and the recipient of the message, respectively
 (PU if $x,y=p$, SU if $x,y=s$), and $Z\in\{\mathrm{A},\mathrm{I}\}$ denotes
the action of the SU ($\mathrm{A}$ if the SU is active and it accesses the channel,
 $\mathrm{I}$ if the SU remains idle).
For example, $q_{ps}^{(\mathrm{A})}$
is the probability that the  PU message is in outage at SUrx, when SUtx transmits.
}
\begin{align}
&q_{pp}^{(\mathrm{I})}(R_{p})\triangleq\mbox{Pr}\left(R_{p}>C\left(\gamma_{p}\right)\vphantom{R_{p}>C\left(\frac{\gamma_{p}}{1+\gamma_{sp}}\right)}\right),
\nonumber\\&
q_{pp}^{(\mathrm{A})}(R_{p})\triangleq\mbox{Pr}\left(R_{p}>C\left(\frac{\gamma_{p}}{1+\gamma_{sp}}\right)\right)\label{PUoutage2},
\end{align}
where $R_p$ denotes the PU transmission rate, measured
in bits/s/Hz,
$C(x)\triangleq\log_{2}(1+x)$ is the (normalized) capacity
of the Gaussian channel with SNR $x$ at the receiver~\cite{Cover}.
This outage definition, as well as the ones introduced later on, assume the use of Gaussian signaling and capacity-achieving
 coding with sufficiently long codewords.
However, our analysis can be extended to include practical codes by
computing the outage probabilities for the specific code considered.
In~(\ref{PUoutage2}), it is assumed that 
 SU transmissions are treated as background Gaussian noise by the PU. This is a reasonable assumption in CRs in which the PU is oblivious to the presence of SUs. 
In general,
 we have $q_{pp}^{(\mathrm{A})}(R_{p})\geq q_{pp}^{(\mathrm{I})}(R_{p})$,
where equality holds if and only if $\gamma_{sp}\equiv 0$ deterministically.
We denote the expected PU throughput
accrued in each time-slot, when the SU is idle and accesses the channel,
as $T_{p}^{(\mathrm{I})}(R_p)=R_p[1-q_{pp}^{(\mathrm{I})}(R_p)]$
 and $T_{p}^{(\mathrm{A})}(R_p)=R_p[1-q_{pp}^{(\mathrm{A})}(R_p)]$, respectively.
\subsection{Operation of the SU}
\label{subsec:SU} 
Unlike the PU
 that uses a simple Type-I Hybrid ARQ mechanism,
it is assumed that the SU uses "best effort" transmission.
Moreover, the SU is provided with side-information about the PU, \emph{e.g.},
ARQ deadline $D$, PU codebook and feedback information from PUrx (ACK/NACK messages). 
 This is consistent with the common characterization
of the PU as a legacy system,
 and of the SU as an opportunistic and cognitive system,
which exploits the primary ARQ feedback to create a best-effort link with maximized throughput,
 while the flow control mechanisms are left to the upper layers.
By overhearing the feedback information from PUrx,
the SU can thus track the primary ARQ state $t$. 
Moreover, by leveraging the PU codebook,
 SUrx attempts, in any time-slot, to decode the PU message, which enables the following IC techniques at SUrx:
\begin{itemize}
 \item  \emph{Forward IC} (FIC): by decoding the PU message, SUrx can
perform IC in the current as well as 
in the following ARQ retransmissions, if these occur,
 to achieve a larger SU throughput;
\item \emph{Backward IC} (BIC):
 SUrx buffers the received signals corresponding to 
SU transmissions which undergo outage due to severe interference from the PU.
These transmissions can later be recovered
using IC on the buffered received signals,
if the interfering PU message is successfully decoded by SUrx in a subsequent primary ARQ retransmission attempt.
\end{itemize}
We define the \emph{SU buffer state} $b\in\mathbb N(0,B)$
as the number of received signals currently buffered at SUrx,
where $B\in\mathbb N(0,D-1)$\footnote{
Note that $B\leq D-1$, since the same PU message is transmitted at most $D$ times by PUtx.
Once the ARQ deadline $D$ is reached, a new PU transmission occurs, and the buffer is emptied.
} denotes the buffer size.
Moreover, we define the \emph{PU message knowledge state} $\Phi\in\{\K ,\U \}$, which denotes the knowledge at SUrx about the 
PU message currently handled by the PU. Namely, if $\Phi=\K $,
 then SUrx knows the PU message, thus enabling FIC/BIC;
 conversely ($\Phi=\U $), the PU message is unknown to SUrx.
\begin{remark}[Feedback Information]
 Note that PUrx needs to report one feedback bit to inform PUtx (and the SU, which overhears
the feedback) on the transmission outcome (ACK/NACK).
On the other hand, two feedback bits need to be reported by SUrx to SUtx:
one bit to inform SUtx as to whether the PU message has been successfully decoded,
so that SUtx can track the PU message knowledge state $\Phi$;
and one bit to inform SUtx as to whether the received signal has been buffered,
so that SUtx can track the SU buffer state $b$.
Herein, we assume ideal (error-free) feedback channels,
so that the SU can track 
$(t,b,\Phi)$, and the PU can track the ARQ state $t$.
However, optimization is possible with imperfect
 observations as well~\cite{Levonoisyobs}.
\qed
\end{remark}

We now further detail the operation of the SU
for $\Phi~\in~\{\K ,\U \}$.
\subsubsection{PU message unknown to SUrx ($\Phi=\U $)}
When $\Phi=\U $ and the 
 SU is idle, SUrx attempts to decode the PU message, so as to enable FIC/BIC.
A decoding failure occurs if the rate of the PU message, $R_p$,
exceeds the capacity of the channel PUtx$\rightarrow$SUrx, with SNR $\gamma_{ps}$.
We denote the corresponding outage probability as
$q_{ps}^{(\mathrm{I})}(R_{p})=\mathrm{Pr}(R_p>C(\gamma_{ps}))$.

\begin{figure}[t]
\centering
\includegraphics[width=\linewidth,trim = 30mm 0mm 30mm 0mm,clip=true]{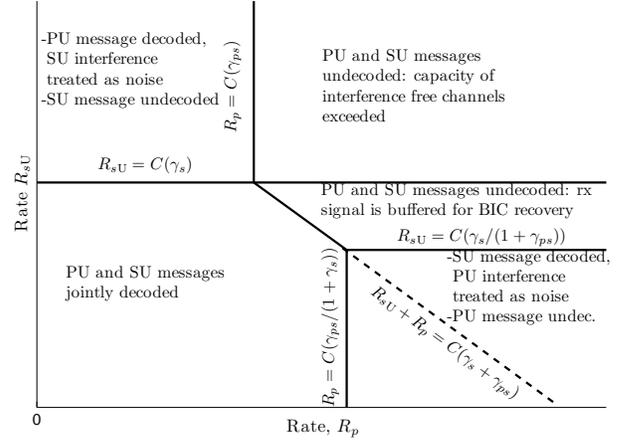}
\caption{Decodability regions for PU message (rate $R_p$) and SU message (rate $R_{s\U }$) at SUrx, for a fixed SNR pair $(\gamma_s,\gamma_{ps})$}
\label{fig:MACfig}
\end{figure}

If the SU accesses the channel, SU transmissions are performed with rate $R_{s\U }$ (bits/s/Hz)
and are interfered by the PU.
 SUrx thus attempts to decode both
the SU and PU messages;
moreover, if the decoding of the SU message fails due to severe interference from the PU,
the received signal is buffered for future BIC recovery.
Using standard information-theoretic results~\cite{Cover},
with the help of Fig.~\ref{fig:MACfig},
we define the following SNR regions associated with the decodability of the 
SU and PU messages at SUrx,
where $\mathcal A^c$ denotes the complementary set of $\mathcal A$:\footnote{Herein, we assume optimal joint decoding techniques of 
the SU and PU messages. Using other techniques, \emph{e.g.}, successive IC,
the SNR regions may change accordingly, without providing any further insights in the
following analysis.}
\begin{align}
 \label{r1}
& \Gamma_{\mathrm{p}}(R_{s\U },R_p)\triangleq 
\left\{\vphantom{\frac{x}{x}}
\left(\gamma_s,\gamma_{ps}\right):
 R_{s\U }\leq C\left(\gamma_{s}\right),
 R_p\leq C\left(\gamma_{ps}\right),\right.\nonumber\\&
\qquad\qquad\qquad\qquad
\left.\vphantom{\frac{x}{x}}
 R_{s\U }+R_p\leq C\left(\gamma_{s}+\gamma_{ps}\right)
\right\}
\\\label{r2}&
\quad
\bigcup
\left\{\left(\gamma_s,\gamma_{ps}\right):
 R_{s\U }>C\left(\gamma_{s}\right),
 R_p\leq C\left(\frac{\gamma_{ps}}{1+\gamma_{s}}\right)\right\}
\end{align}
\begin{align}
\label{r3}
&  \Gamma_{\mathrm{s}}(R_{s\U },R_p)\triangleq 
\left\{\vphantom{\frac{x}{x}}
\left(\gamma_s,\gamma_{ps}\right):
 R_{s\U }\leq C\left(\gamma_{s}\right),
 R_p\leq C\left(\gamma_{ps}\right),\right.\nonumber\\&
\qquad\qquad\qquad\qquad
\left.\vphantom{\frac{x}{x}}
 R_{s\U }+R_p\leq C\left(\gamma_{s}+\gamma_{ps}\right)
\right\}\\\label{r4}&
\quad\bigcup
\left\{\left(\gamma_s,\gamma_{ps}\right):
 R_p>C\left(\gamma_{ps}\right),
 R_{s\U }\leq C\left(\frac{\gamma_{ps}}{1+\gamma_{s}}\right)\right\}
\\\label{r5}
&\Gamma_{\mathrm{buf}}(R_{s\U },R_p)\triangleq 
\left\{\vphantom{\frac{x}{x}}\Gamma_{\mathrm{p}}(R_{s\U },R_p)\cup \Gamma_{\mathrm{s}}(R_{s\U },R_p)\right\}^c
\\&\nonumber\qquad
\bigcap
\left\{\vphantom{\frac{x}{x}}\left(\gamma_s,\gamma_{ps}\right):
R_{s\U }\leq C\left(\gamma_{s}\right)
\right\}.
\end{align}
The SNR regions~(\ref{r1}) and~(\ref{r3}) guarantee that the two rates $R_p$ and $R_{s\U }$
are within the multiple access channel region formed by the two transmitters (PUtx and SUtx)
and SUrx~\cite{Cover}, so that both the SU and PU messages are correctly decoded via joint
decoding techniques.
On the other hand, in the SNR region~(\ref{r4}) (respectively,~(\ref{r2})), 
only the SU (PU) message is successfully decoded
at SUrx by treating the
interference from the PU (SU) as background noise.
If the SNR pair falls outside the two regions~(\ref{r3}) and~(\ref{r4})
 (respectively,~(\ref{r1}) and~(\ref{r2})), then SUrx incurs a failure in decoding
the SU (PU) message.
Therefore,
when $(\gamma_s,\gamma_{ps})\in\Gamma_{\mathrm{s}}(R_{s\U },R_p)$,
 SUrx successfully decodes the SU message.
The corresponding expected SU throughput is thus given by
\begin{align}
T_{s\U }(R_{s\U },R_{p})\triangleq R_{s\U }\mathrm{Pr}\left((\gamma_s,\gamma_{ps})\in
\Gamma_{\mathrm{s}}(R_{s\U },R_p)\right).
 \end{align}
Similarly, 
when $(\gamma_s,\gamma_{ps})\in\Gamma_{\mathrm{p}}(R_{s\U },R_p)$,
 SUrx successfully decodes the PU message. 
We denote the corresponding outage probability
as $q_{ps}^{(\mathrm{A})}(R_{s\U },R_{p})\triangleq\mathrm{Pr}
\left((\gamma_s,\gamma_{ps})\notin\Gamma_{\mathrm{p}}(R_{s\U },R_p)\right)$.
Note that
$q_{ps}^{(\mathrm{A})}(R_{s\U },R_{p})>q_{ps}^{(\mathrm{I})}(R_{p})$,
since SU transmissions interfere with the decoding of the PU message.

Finally, in~(\ref{r5}), 
 the decoding of both the SU and PU messages fails,
since the SNR pair $(\gamma_s,\gamma_{ps})$ falls outside both regions
$\Gamma_{\mathrm{p}}(R_{s\U },R_p)$ and $\Gamma_{\mathrm{s}}(R_{s\U },R_p)$.
 However, the rate $R_{s\U }$ is within the capacity region of the interference
free channel ($R_{s\U }\leq C\left(\gamma_{s}\right)$),
so that
the SU message can be recovered via BIC, should the PU message become available in a future ARQ 
retransmission attempt.
The received signal is thus buffered at SUrx. 
 We denote the \emph{buffering probability} as
\begin{align}\label{omega}
p_{s,\mathrm{buf}}(R_{s\U },R_{p})&\triangleq
\mathrm{Pr}\left(\left(\gamma_s,\gamma_{ps}\right)\in\Gamma_{\mathrm{buf}}(R_{s\U },R_p)\right)
\nonumber\\&=
\mathrm{Pr}\left((\gamma_s,\gamma_{ps})\in
\Gamma_{\mathrm{s}}(R_{s\U },0)\right)
\\&\nonumber\quad
-\mathrm{Pr}\left((\gamma_s,\gamma_{ps})\in
\Gamma_{\mathrm{s}}(R_{s\U },R_p)\right)>0,
\end{align}
where the second equality follows from inspection of Fig.~\ref{fig:MACfig}.
\subsubsection{PU message known to SUrx ($\Phi=\K $)}
When $\Phi=\K $, SUrx performs FIC on the received signal,
thus enabling interference free SU transmissions.
The SU transmits with rate $R_{s\K }$, and
the accrued throughput is given by
$T_{s\K }(R_{s\K })=R_{s\K }\mathrm{Pr}\left(R_{s\K }<C(\gamma_{s})\right)$.

We now provide an example to illustrate the use of FIC/BIC at SUrx.
\begin{example}
 Consider a sequence of $3$ primary retransmission attempts
 in which the SU always accesses the channel.
Initially, the PU message is unknown to SUrx, hence the PU message knowledge state is set to
$\Phi=\U $ in the first time-slot,
 and the SU transmits with rate $R_{s\U }$.
Assume that the SNR pair
 $(\gamma_s(1),\gamma_{ps}(1))$ falls in $\Gamma_{\mathrm{buf}}(R_{s\U },R_p)$.
 Then, neither the SU nor the PU messages are successfully decoded by SUrx,
but the received signal is buffered for future BIC recovery.
In the second time-slot, $(\gamma_s(2),\gamma_{ps}(2))\in\Gamma_{\mathrm{s}}(R_{s\U },R_p)\cap\Gamma_{\mathrm{p}}(R_{s\U },R_p)$,
hence both the SU and PU messages are correctly decoded by SUrx,
and the PU message knowledge state switches to $\Phi=\K $. 
At this point, SUrx performs BIC on the previously buffered received signal to recover
the corresponding SU message.
In the third time-slot, SUtx transmits with rate $R_{s\K }$,
 and decoding at SUrx takes place after cancellation of the interference from the PU via FIC.\qed
\end{example}

We now briefly elaborate on the choice of the transmission rate
 $R_{s\K }$.
 Since its value does
not affect the outage performance at PUrx~(\ref{PUoutage2}) and
the evolution of the ARQ process,
$R_{s\K }$ is chosen so as to maximize $T_{s\K }(R_{s\K })$.
 Therefore, from~(\ref{omega}) we obtain
\begin{align}\label{tau}
 T_{s\K }(R_{s\K })\geq &
T_{s\K }(R_{s\U })=
T_{s\U }(R_{s\U },R_{p})
\nonumber\\&
+p_{s,\mathrm{buf}}(R_{s\U },R_{p})R_{s\U }>T_{s\U }(R_{s\U },R_{p}).
\end{align}
Conversely, the choice of the rate $R_{s\U }$ is not as straightforward, 
since its value reflects a trade-off between 
 the potentially larger throughput accrued with a larger rate $R_{s\U }$
 and the corresponding diminished capabilities for IC caused
 by the more difficult decoding of the PU message by SUrx.

In the following treatment, the rates $R_{s\K }$, $R_{s\U }$ and $R_p$
are assumed to be fixed parameters of the system, and they are not considered part of 
the optimization  (see Sec.~\ref{sec:numres} for further elaboration in this regard).
For the sake of notational convenience, we omit
the dependence of the quantities defined above
 on them.
Moreover, 
for clarity, we consider the case $B=D-1$
in which SUrx can buffer up to $D-1$ received signals. 
However, the following analysis can be extended to a generic value of $B$.
\section{Policy Definition and Optimization Problem}
\label{sec:performance}
We model the evolution of the network as a Markov Decision Process~\cite{Bertsekas,DJWhite}.
Namely,
we denote the state of the PU-SU system by the tuple $(t,b,\Phi)$,
where $t\in\mathbb N(1,D)$ is the primary ARQ state,
$b\in\mathbb N(0,B)$ is the SU buffer state and
 $\Phi\in\{\U ,\K \}$ is the PU message knowledge state.
$(t,b,\Phi)$
takes values in the state space $\mathcal S\equiv\mathcal S_{\U }\cup\mathcal S_{\K }$,
where
$\mathcal S_{\K }\equiv\{(t,0,\K ):t\in\mathbb N(2,D)\}$
and
$\mathcal S_{\U }\equiv\{(t,b,\U ):t\in\mathbb N(1,D),b\in\mathbb N(0,t-1)\}$
are the sets of states where the PU message is known and unknown to SUrx, respectively.

The SU follows a \emph{stationary randomized access policy} $\mu\in\mathcal U\equiv\left\{\mu:\mathcal{S}\mapsto[0,1]\right\}$,
which determines the secondary access probability for each state $\mathbf s\in\mathcal S$.
Note that, from~\cite{Ross1989}, this choice is without loss of optimality for the specific problem at hand.
Namely, in state $(t,b,\Phi)\in\mathcal S$, the SU is "active", \emph{i.e.}, it accesses the channel, with probability
$\mu(t,b,\Phi)$ and stays "idle" with probability $1-\mu(t,b,\Phi)$.
We denote the "active" and "idle" actions as $\mathrm{A}$ and $\mathrm{I}$, respectively.

With these definitions at hand,
we define the following average long-term metrics under $\mu$:
the SU throughput $\bar T_s(\mu)$, the SU power expenditure $\bar P_s(\mu)$
and the PU throughput $\bar T_p(\mu)$, given by
\begin{align}
\bar T_s(\mu)=&\lim_{N\to+\infty}\frac{1}{N}\mathbb E\left[\left.\sum_{n=0}^{N-1}
R_{s\Phi_n}\mathbf 1\left(\{Q_n=\mathrm{A}\}\cap O_{s,n}^c\right)\right|\mathbf s_0
\right]\nonumber\\\label{Ts0}
&+\lim_{N\to+\infty}\frac{1}{N}\mathbb E\left[\left.\sum_{n=0}^{N-1}
R_{s\U}B_n\mathbf 1(O_{ps,n}^c)\right|\mathbf s_0
\right],\\
\label{Ps0}
\bar P_s(\mu)=&P_s
\lim_{N\to+\infty}\frac{1}{N}\mathbb E\left[\left.\sum_{n=0}^{N-1}\mathbf 1\left(\{Q_n=\mathrm{A}\}\right)
\right|\mathbf s_0
\right]
,\\
\label{Tp0}\bar T_p(\mu)=&
\lim_{N\to+\infty}\frac{1}{N}\mathbb E\left[\left.\sum_{n=0}^{N-1}
R_p\mathbf 1\left(O_{p,n}^c\right)\right|\mathbf s_0\right],
\end{align}
where $n$ is the time-slot index,
$\mathbf s_0\in\mathcal S$ is the initial state in time-slot $0$;
 $\Phi_n\in\{\K ,\U \}$ is the  PU message knowledge state and
$B_n$ is the SU buffer state in time-slot $n$;
$Q_n\in\{\mathrm{A},\mathrm{I}\}$ is the action of the SU, drawn according to the access policy
$\mu$; $O_{s,n}$ and $O_{ps,n}$ denote the outage events at SUrx for the decoding of the SU and PU messages, so that
$O_{s,n}^c$ and $O_{ps,n}^c$ denote successful decoding
of the SU and PU messages by SUrx, respectively;
$O_{p,n}$ denotes the outage event at PUrx, so that $O_{p,n}^c$ denotes successful decoding
of the PU message by PUrx;
and $\mathbf 1(E)$ is the indicator function of the event $E$.
 Note that all the quantities defined above are independent of the initial state 
$\mathbf s_0$. In fact,
 starting from any $\mathbf s_0\in\mathcal S$, the system reaches
with probability $1$
the positive recurrent state $(1,0,\U )$ (new PU transmission)
 within a finite number of time-slots,
due to the ARQ deadline. Due to the Markov property, from this state on,
 the evolution of the process is independent of the initial transient behavior,
which has no effect on the time averages defined 
in~(\ref{Ts0}),~(\ref{Ps0}) and~(\ref{Tp0}).

In this work, we study the problem
 of maximizing the average long-term SU throughput subject to constraints on the average long-term PU throughput loss
 and SU power.
 Specifically,
\begin{align}\label{opt}
\mu^{*}=
\arg\max_{\mu}\bar{T}_{s}(\mu)
 \mathrm{\ s.t.\ }
&\bar{T}_{p}(\mu)\geq T_{p}^{(\mathrm{I})}(1-\epsilon_{\mathrm{PU}}),
\nonumber
\\&
\bar{P}_{s}(\mu)\leq\mathcal{P}_{s}^{(\mathrm{th})},
\end{align}
 where $\epsilon_{\mathrm{PU}}\in[0,1]$
and $\mathcal{P}_{s}^{(\mathrm{th})}\in[0,P_{s}]$ represent the (normalized)
maximum tolerated  PU throughput loss 
 with respect to the case in which the SU is idle
and the SU power constraint, respectively.
This problem entails a trade-off in the operation of the SU.
On the one hand, the SU is incentivized to transmit in order to 
increase its throughput and
 to optimize the buffer occupancy at SUrx (\emph{i.e.}, failed SU transmissions
which are potentially recovered via BIC).
On the other hand, SU transmissions might
jeopardize the correct decoding of the PU message at SUrx, thus impairing the use of FIC/BIC,
and might violate the constraints in~(\ref{opt}).

Under $\mu\in\mathcal{U}$, 
the state process is a stationary Markov chain, with steady state distribution
 $\pi_\mu$~\cite{Kemeny1960,DJWhite}.
$\pi_\mu(\mathbf s),\mathbf s\in\mathcal S$, is the long-term fraction of the time-slots spent
in state $\mathbf s$, \emph{i.e.},
$
 \pi_\mu(\mathbf s)=\underset{N\to+\infty}{\lim}\frac{1}{N}\sum_{n=0}^{N-1}\mathrm{Pr}_\mu^{(n)}\left(\mathbf s|\mathbf s_0\right)
 $,
 where $\mathrm{Pr}_\mu^{(n)}\left(\mathbf s|\mathbf s_0\right)$ is the $n$-step transition probability of the chain
from state $\mathbf s_0$.\footnote{Similarly to
(\ref{Ts0}),~(\ref{Ps0}) and~(\ref{Tp0}), $\pi_\mu(\mathbf s)$
is independent of the initial state $\mathbf s_0$, due to the recurrence of 
state $(1,0,\U )$.}
In state $(t,b,\U )$, the SU accesses the channel with probability
$\mu\left(t,b,\U \right)$, thus accruing
 the throughput $\mu\left(t,b,\U \right)T_{s\U }$.
Moreover, if SUrx successfully decodes the PU message 
(with probability $
1-q_{ps}^{(\mathrm{I})}
-\mu(t,b,\U )(q_{ps}^{(\mathrm{A})}-q_{ps}^{(\mathrm{I})})$),
$bR_{s\U }$ bits are recovered by performing BIC on the buffered received signals,
 yielding an additional BIC throughput.
Similarly, in state $(t,0,\K )$, the SU accrues the throughput $\mu\left(t,0,\K \right)T_{s\K }$.
Then, we can rewrite~(\ref{Ts0}) and~(\ref{Ps0}) in terms of the
steady state distribution and of the cost/reward in each state as
\begin{align}\label{Ts}
\!\!\!\!\bar T_s(\mu)\!=\!T_{s\U }\bar W_s(\mu)\!+\!\bar F_s(\mu)\!+\!\bar B_s(\mu),\ 
\bar P_s(\mu)\!=\!P_s\bar W_s(\mu),
\end{align}
where the \emph{SU access rate} $\bar W_s(\mu)$, 
\emph{i.e.}, the average long-term number of secondary channel accesses per time-slot,
the \emph{FIC throughput} $\bar F_s(\mu)$ and the \emph{BIC throughput} $\bar B_s(\mu)$
are defined as
\begin{align}\label{Ws}
\begin{array}{rl}
\bar W_s(\mu)\triangleq&\!\!\!\sum_{\mathbf s\in\mathcal S}\pi_\mu\left(\mathbf s\right)\mu\left(\mathbf s\right),\\
\bar F_s(\mu)\triangleq&\!\!\!\sum_{t=2}^D\pi_\mu\left(t,0,\K \right)
\mu\left(t,0,\K \right)(T_{s\K }-T_{s\U }),\\
\bar B_s(\mu)\triangleq&\!\!\!
\sum_{t=1}^D\sum_{b=0}^{t-1}\pi_\mu\left(t,b,\U \right)bR_{s\U }
\\&\times
\left[
1-q_{ps}^{(\mathrm{I})}
-\mu\left(t,b,\U \right)\left(q_{ps}^{(\mathrm{A})}-q_{ps}^{(\mathrm{I})}\right)\right].
\end{array}
\end{align}
In~(\ref{Ts}), $T_{s\U }\bar W_s(\mu)$ is the SU throughput
 attained without FIC/BIC, while
 the terms $\bar F_s(\mu)$ and $\bar B_s(\mu)$ account for the throughput gains 
of FIC and BIC, respectively.
Conversely, the PU accrues the throughput $T_{p}^{(\mathrm{I})}$
if the SU is idle and $T_{p}^{(\mathrm{A})}$ if the SU accesses the channel,
so that~(\ref{Tp0}) is given by
\begin{align}\label{Tp}
\bar T_p(\mu)=T_{p}^{(\mathrm{I})}-(T_{p}^{(\mathrm{I})}-T_{p}^{(\mathrm{A})})\bar W_s(\mu).
\end{align}
The quantity $(T_{p}^{(\mathrm{I})}-T_{p}^{(\mathrm{A})})\bar W_s(\mu)$ is referred to
as the \emph{PU throughput loss} induced by the secondary access policy $\mu$~\cite{IT_ARQ}. 
The following result follows directly from~(\ref{opt}),~(\ref{Ts}) and~(\ref{Tp}).
\begin{lemma}
The problem~(\ref{opt}) is equivalent to
\begin{align}
\label{refopt2}
\mu^{*}&=
 {\arg\max}_{\mu\in\mathcal{U}}\bar{T}_s(\mu)
\\&\nonumber
\mathrm{s.t.\  }\bar{W}_s(\mu)
\leq\min\left\{ \frac{(1-q_{pp}^{(\mathrm{I})})
\epsilon_{\mathrm{PU}}}{q_{pp}^{(\mathrm{A})}-q_{pp}^{(\mathrm{I})}},
\frac{\mathcal{P}_{s}^{(\mathrm{th})}}{P_{s}}\right\} \triangleq\epsilon_{\mathrm{W}}.
\end{align}
\qed
\end{lemma}
In the next section, we characterize the solution of~(\ref{refopt2}).
We will need the following definition.
\begin{definition}\label{threshold}
Let $\mu$ be the policy such that
secondary access takes place if and only if the PU message is known to SUrx,
\emph{i.e.},
 $\mu(\mathbf s)=1,\ \forall\mathbf s\in\mathcal S_{\K }$,
$\mu(\mathbf s)=0,\ \forall\mathbf s\in\mathcal S_{\U }$.
We denote the SU access rate achieved by such policy as
 $\epsilon_{\mathrm{th}}=\bar W(\mu)$.
The system is in the \emph{low SU access rate regime}
if $\epsilon_{\mathrm{W}}\leq \epsilon_{\mathrm{th}}$ in~(\ref{refopt2}).
Otherwise, the system is in the \emph{high SU access rate regime}.
\qed
\end{definition}
\section{Optimal Policy}\label{sec:SU_tx_policy}
In this section, we
characterize in closed form the optimal policy
in the low SU access rate regime, and we
 present an algorithm to derive the optimal policy
in the high SU access rate regime.

\subsection{Low SU Access Rate Regime}\label{lsar}
The next lemma shows that, in the low SU access rate regime,
 an optimal policy prescribes that secondary access only takes place in the 
states where the PU message is known to SUrx,
 with an equal probability in all such states.
 It follows that only FIC, and not BIC, is needed in this regime to attain
 optimal performance.
\begin{lemma}
\label{optpollow}
In the low SU access rate regime $\epsilon_{\mathrm{W}}\leq \epsilon_{\mathrm{th}}$,
 an optimal policy is given by\footnote{The optimal policy in the low SU access rate
is not unique.
In fact, any policy $\mu$ such that $\mu(\mathbf s)=0,\ \forall\mathbf s\in\mathcal S_{\U}$
and $\bar W_s(\mu)=\epsilon_{\mathrm{th}}$ is optimal, attaining
the same throughput $\bar T_s(\mu)=T_{s\K}\epsilon_{\mathrm{th}}$ as (\ref{optimalforlow}).}
  \begin{align}\label{optimalforlow}
 \mu^{*}(\mathbf s)=\frac{\epsilon_{\mathrm{W}}}{\epsilon_{\mathrm{th}}},\ \forall\mathbf s\in\mathcal S_{\K },
\ \mu^{*}(\mathbf s)=0,\ \forall\mathbf s\in\mathcal S_{\U }.
  \end{align}
Moreover,
$\bar T_s(\mu^*)=T_{s\K }\epsilon_{\mathrm{W}}$,
 $\bar P_s(\mu^*)=P_s\epsilon_{\mathrm{W}}$,
and\\
$\bar T_p(\mu^*)
=T_{p}^{(\mathrm{I})}
-(T_{p}^{(\mathrm{I})}-T_{p}^{(\mathrm{A})})\epsilon_{\mathrm{W}}$.
\qed
\end{lemma}
\begin{proof}
For any policy $\mu\in\mathcal U$ obeying the 
SU access rate
constraint $\bar W_s(\mu)\leq \epsilon_{\mathrm{W}}$,
we have
$\bar T_s(\mu)\leq\bar W_s(\mu)T_{s\K }\leq \epsilon_{\mathrm{W}}T_{s\K }$.
The first inequality holds since
$\bar W_s(\mu)T_{s\K }$
 is the long-term throughput
achievable when the PU message is known \emph{a priori} at SUrx, which is an upper bound
to the performance;
the second from the SU access rate constraint.
The upper bound
$\epsilon_{\mathrm{W}}T_{s\K }$
is achieved by policy~(\ref{optimalforlow}), as can be directly seen
by substituting~(\ref{optimalforlow}) in~(\ref{Ts}),~(\ref{Ws}).
\end{proof}
\begin{remark}
 Note that secondary accesses in states $\mathcal S_{\U }$,  where the PU message is unknown to SUrx,
would obtain a smaller throughput, namely
at most
 $T_{s\U }+p_{s,\mathrm{buf}} R_{s\U }\leq T_{s\K }$, where
$T_{s\U }$  is the "instantaneous" throughput and
$p_{s,\mathrm{buf}} R_{s\U }$ is the BIC throughput,
 \emph{possibly} recovered via BIC in a future ARQ retransmission.
Therefore, SU accesses in states $\mathcal S_{\K }$
are more "cost effective".
\qed
\end{remark}

\subsection{High SU Access Rate Regime}
In this section, we study the high SU access rate regime
 in which $\epsilon_{\mathrm{W}}>\epsilon_{\mathrm{th}}$,
 thus complementing the analysis above for the regime where $\epsilon_{\mathrm{W}}\leq\epsilon_{\mathrm{th}}$.
 It will be seen that, if $\epsilon_{\mathrm{W}}>\epsilon_{\mathrm{th}}$,
 unlike in the low SU access rate regime,
 the SU should generally access the channel also in 
states $\mathcal S_{\U }$ where the PU message is unknown to SUrx
in order to achieve the optimal performance.
 Therefore, both BIC and FIC are necessary to attain optimality.
In this section, we derive the optimal policy.
We first introduce some necessary definitions and notations.
\begin{definition}[Secondary access efficiency]
\label{definitioneta}
We define the \emph{secondary access efficiency} under policy $\mu\in\mathcal U$ in state $\mathbf s\in\mathcal S$ as
\begin{align}
\eta_\mu\left(\mathbf s\right)=
\frac{\frac{\mathrm d\bar T_s(\mu)}{\mathrm d\mu(\mathbf s)}}{\frac{\mathrm d\bar W_s(\mu)}{\mathrm d\mu(\mathbf s)}}. 
\end{align}
\qed
\end{definition}
The secondary access efficiency can be interpreted as follows.
If the secondary access probability is increased in state $\mathbf s\in\mathcal S$
by a small amount $\delta$,
then the PU throughput loss is increased by an amount equal to
$\delta (T_{p}^{(\mathrm{I})}-T_{p}^{(\mathrm{A})})\frac{\mathrm d\bar W_s(\mu)}{\mathrm d\mu(\mathbf s)}$
(from~(\ref{Tp})),
the SU power is increased by an amount equal to $\delta P_s\frac{\mathrm d\bar W_s(\mu)}{\mathrm d\mu(\mathbf s)}$
(from~(\ref{Ts})), and the SU throughput augments or diminishes 
 by an amount equal to $\delta\frac{\mathrm d\bar T_s(\mu)}{\mathrm d\mu(\mathbf s)}$ (depending on the sign of
the derivative).
Therefore, $\eta_\mu\left(\mathbf s\right)$
yields the rate of increase
(or decrease if $\eta_\mu\left(\mathbf s\right)<0$) of the SU throughput
 per unit increase of the SU access rate, as
induced by augmenting the secondary channel access probability in state $\mathbf s$.
Equivalently, it measures
 how efficiently the SU can access the channel in state $\mathbf s$,
in terms of maximizing the SU throughput gain while minimizing its negative impact
 on the PU throughput and on the SU power expenditure.

\begin{remark}
It is worth noting that the definition of $\eta_\mu\left(\mathbf s\right)$ given in
Def.~\ref{definitioneta} is not completely rigorous.
In fact, under a generic policy $\mu$, the Markov chain of the
PU-SU system may not be irreducible~\cite{Kemeny1960}, so that 
 state $\mathbf s$ may not be accessible, hence
 $\pi_\mu(\mathbf s)=0$ and $\frac{\mathrm d\bar T_s(\mu)}{\mathrm d\mu(\mathbf s)}
=
\frac{\mathrm d\bar W_s(\mu)}{\mathrm d\mu(\mathbf s)}=0$.
One example is the idle policy $\mu(\mathbf s)=0,\ \forall \mathbf s$:
since the SU is always idle, the buffer at SUrx is always empty,
hence states $(t,b,\U )$ with $b>0$ are never accessed.
To overcome this problem, a formal definition is given in App. \ref{rigeta}, 
by treating the Markov chain of the PU-SU system as the limit of an
 irreducible Markov chain.
$\eta_\mu\left(\mathbf s\right)$ is explicitly derived
in Lemma~\ref{computeeta} in App.~\ref{rigeta}.
\qed
\end{remark}

We denote the indicator function of state $\mathbf s$ as 
 $\delta_{\mathbf s}:\mathcal S\mapsto \{0,1\}$, with
 $\delta_{\mathbf s}(\mathbf{\mathbf s})=1$,
 $\delta_{\mathbf s}(\mathbf \sigma)=0,\ \forall\mathbf\sigma\neq\mathbf s$.
Moreover, we denote the policy at the $i$th  iteration of the algorithm as $\mu^{(i)}$.
 We are now ready to describe the algorithm that obtains an optimal policy
 in the high SU access rate regime.
 An intuitive explanation of the algorithm can be found below.
\begin{algo}[Derivation of the optimal policy]\label{algorithm}$ $\\
{\begin{enumerate}
 \item \textbf{Initialization}:
\begin{itemize}
\item Let $\mu^{(0)}$ be the policy
$\mu^{(0)}(\mathbf s)=0,\ \forall\  \mathbf s\in\mathcal S_{\U }$,\\
$\mu^{(0)}(\mathbf s)=1,\ \forall\  \mathbf s\in\mathcal S_{\K }$,
and $i=0$.
\item Let $\mathcal S_{\mathrm{idle}}^{(0)}\equiv\{\mathbf s\in\mathcal S:\mu^{(0)}(\mathbf s)=0\}\equiv\mathcal S_{\U }$ be the set of states where the SU is idle.
\end{itemize}
 \item \textbf{Stage $i$}:
\begin{enumerate}
\item Compute $\eta_{\mu^{(i)}}(\mathbf s)$, $\forall\ \mathbf s\in\mathcal S_{\mathrm{idle}}^{(i)}$
and let $\mathbf s^{(i)}\triangleq\arg\max_{\mathbf s\in\mathcal S_{idle}^{(i)}}\eta_{\mu^{(i)}}(\mathbf s)$.
\item If $\eta_{\mu^{(i)}}(\mathbf s^{(i)})\leq 0$, go to step 3). 
Otherwise, let $\mu^{(i+1)}=\mu^{(i)}+\delta_{\mathbf s^{(i)}}$,
$\mathcal S_{\mathrm{idle}}^{(i+1)}\!\equiv\!\mathcal S_{\mathrm{idle}}^{(i)}\setminus\left\{\mathbf s^{(i)}\right\}$.
\item Set $i:=i+1$. If $\mathcal S_{\mathrm{idle}}^{(i)}\equiv \emptyset$, go to step 3).
Otherwise, repeat from step 2).
\end{enumerate}
\item Let $N=i$,
the sequence of states $(\mathbf s^{(0)},\dots,\mathbf s^{(N-1)})$
and of policies $(\mu^{(0)},\dots,\mu^{(N-1)})$.
\item \textbf{Optimal policy}: given $\epsilon_{\mathrm{W}}$,
                      \begin{enumerate}
  \item If $\bar W_s(\mu^{(N-1)})\leq \epsilon_{\mathrm{W}}$, then
  $\mu^{*}=\mu^{(N-1)}$.
  \item Otherwise, $\mu^{*}=\lambda \mu^{(j)}+(1-\lambda) \mu^{(j+1)}$,
  where
$j\!\triangleq\!\max\left\{i\!:\!\bar W_s\left(\mu^{(i)}\right)\!\leq\!\epsilon_{\mathrm{W}}\right\}$
and $\lambda\in (0,1]$ uniquely solves
$\bar W_s(\lambda \mu^{(j)}+(1-\lambda) \mu^{(j+1)})=\epsilon_{\mathrm{W}}$.\qed
 \end{enumerate}
\end{enumerate}
}
\end{algo}
The algorithm, starting from the optimal policy for the case $\epsilon_{\mathrm{W}}=\epsilon_{\mathrm{th}}$
(Lemma~\ref{optpollow}),
 ranks the states in the set $\mathcal S_{\U }$ in decreasing order of
secondary access efficiency, and iteratively allocates the secondary access to the 
state with the highest efficiency, among the states where the SU is idle.
The rationale of this step is that
 secondary access in the most efficient state yields the steepest increase of the SU
throughput, per unit increase of the SU access rate or, equivalently, of the PU throughput loss and of the SU
power expenditure. 
 The optimality of Algorithm~\ref{algorithm}
is established in the following theorem.
\begin{thm}\label{algooptimal}
Algorithm~\ref{algorithm} returns an optimal policy for the optimization problem~(\ref{refopt2}).
\qed
\end{thm}
\begin{proof}
See App.~\ref{proofofalgo}.
\end{proof}

\section{Special Case: degenerate cognitive radio network scenario}\label{special}

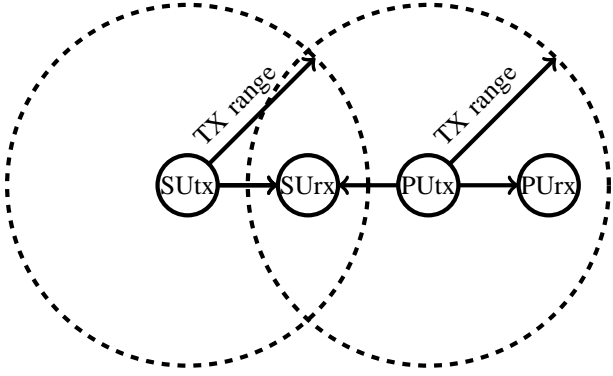
\begin{figure}
    \centering
    {
\begin{tikzpicture}[scale=0.8]
\draw [dashed, ultra thick] (0,0) circle [radius=3];
\draw [dashed, ultra thick] (4,0) circle [radius=3];
\draw [ultra thick, ->] (0,0) -- (2.12,2.12);
\draw [ultra thick, ->] (4,0) -- (6.12,2.12);
\draw [ultra thick, ->] (0,0) -- (2-0.5,0);
\draw [ultra thick, ->] (0,0) -- (2-0.5,0);
\draw [ultra thick, ->] (4,0) -- (2+0.5,0);
\draw [ultra thick, ->] (4,0) -- (6-0.5,0);
\draw [fill=white, ultra thick] (0,0) circle [radius=0.5];;
\draw [fill=white, ultra thick] (2,0) circle [radius=0.5];;
\draw [fill=white, ultra thick] (4,0) circle [radius=0.5];;
\draw [fill=white, ultra thick] (6,0) circle [radius=0.5];;
\node at (0,0) {SUtx};
\node at (2,0) {SUrx};
\node at (4,0) {PUtx};
\node at (6,0) {PUrx};
\node [above, rotate=45] at (1.06,1.06) {TX range};
\node [above, rotate=45] at (5.06,1.06) {TX range};
%
\end{tikzpicture}}
\caption{Degenerate cognitive radio network}
\label{fig:degenerate_cog_net}
\end{figure}

We point out that Algorithm~\ref{algorithm} determines the optimal policy
for a generic set of system parameters.
However, the resulting optimal policy does not always have a structure that is easily interpreted.
In this section, we consider a special case of the general model
discussed so far,
a \emph{degenerate cognitive radio network},
where the activity of the PU is unaffected
by the transmissions of the SU, \emph{i.e.}, the channel gain between
the SU transmitter and the PU receiver is zero. 


Consider the scenario depicted in
Fig.~\ref{fig:degenerate_cog_net}, where PUrx is outside the transmission
range of SUtx, whereas SUrx is inside the transmission range of both SUtx and SUrx.
In this scenario, the interference produced by SU to PU is negligible.
In contrast, the PU produces significant interference at the SU receiver.
The SU thus potentially benefits by employing
the BIC and FIC mechanisms.
We denote this scenario as a \emph{Degenerate cognitive radio network},
and we model it by assuming that 
the SNR of the interfering link SUtx$\rightarrow$PUrx
is deterministically equal to zero, \emph{i.e.}, $\gamma_{sp}=0$.
From~(\ref{PUoutage2}), we then have
 $q_{pp}^{(\I)}=q_{pp}^{(\A)}\triangleq q_{pp}$, \emph{i.e.},
 the outage performance of the PU is unaffected by the activity of the SU,
and the primary ARQ process is independent of the secondary access policy.
We define 
\begin{align}\label{DeltaS}
\Delta_s\triangleq\frac{T_{s\mathrm{K}}-T_{s\mathrm{U}}-p_{s,\mathrm{buf}}R_{s\mathrm{U}}}{R_{s\mathrm{U}}}.
\end{align}
 From~(\ref{tau}),
it follows that $\Delta_s\geq 0$, with equality if $R_{s\mathrm{U}}=R_{s\mathrm{K}}$.
Therefore, $R_{s\mathrm{U}}\Delta_s$ is the marginal throughput gain accrued
in the states where the PU message is known to SUrx, over
the throughput accrued in
 the states where the PU message is unknown
(instantaneous throughput $T_{s\mathrm{U}}$ plus BIC throughput $p_{s,\mathrm{buf}}R_{s\mathrm{U}}$,
possibly recovered in a future ARQ retransmission).
The following lemma proves that, if
 the marginal throughput gain
 $\Delta_s$ is "small",
the secondary accesses in the high SU access rate regime
in a degenerate cognitive radio network
 are allocated, in order, to the states
in $\mathcal S_{\K }$ (Lemma~\ref{optpollow}),
 then to the idle states $(t,b,\U )$ in $\mathcal S_{\U }$,
giving priority to states with low $b$ and $t$ over states with high $b$ and $t$,
 respectively. 
 An illustrative example
of the optimal policy for this scenario is given in Fig.~\ref{fig:degenerate_cog_net_policy}.
\begin{lemma}\label{mainthm}
In the degenerate cognitive radio network scenario
with $q_{pp}^{(\A)}=q_{pp}^{(\I)}=q_{pp}$,
if
\begin{align}\label{HP}
 \Delta_s<\frac{1-q_{ps}^{(\A)}}{q_{ps}^{(\A)}-q_{ps}^{(\I)}}p_{s,\mathrm{buf}},
\end{align}
 the sequence of policies $(\mu^{(0)},\dots,\mu^{(N-1)})$
returned by Algorithm~\ref{algorithm} is such that,
 $\forall i\in\mathbb N(0,N-1)$,
 \begin{align}\label{struc}
\mu^{(i)}(\mathbf s)=&1,\ \forall \mathbf s\in\mathcal S_{\mathrm{K}},
\\
\mu^{(i)}(t,b,\mathrm{U})=&
  \left\{
  \begin{array}{ll}
   1 &b<b^{(i)}(t)\\
   0 &b\geq b^{(i)}(t),
  \end{array}
  \right.,\ \forall \mathbf (t,b,\U)\in\mathcal S_{\mathrm{U}},
 \end{align}
where $b^{(i)}(t)$ is non-increasing in $t$
and non-decreasing in $i$,
with $b^{(0)}(t)=0$ and $b^{(N-1)}(t)=\bar b_{\max}(t)$, \emph{i.e.},
\begin{align}
\label{incri}
&\bar b_{\max}(t)=b^{(N-1)}(t)\geq\dots\geq b^{(i)}(t)\geq b^{(i-1)}(t)
\\\nonumber&\quad\geq\dots\geq b^{(0)}(t)=0.
\\
\label{incrt}
&b^{(i)}(1)\!\geq\!b^{(i)}(2)\!\geq\!\dots\!\geq\!b^{(i)}(t-1)\!\geq\!b^{(i)}(t)\!\geq\!\dots\!\geq\!b^{(i)}(D),
\end{align}
where
   \begin{align}\label{bmax}
\bar b_{\max}(t)=\left\lceil\frac{
\begin{array}{l}
 \frac{T_{s\U}}{R_{s\U}}
\left[1-q_{pp}\left(q_{ps}^{(\A)}-q_{ps}^{(\I)}\right)A_0(t+1)\right]\\
+
\left(\frac{1-q_{ps}^{(\A)}}
{q_{ps}^{(\A)}-q_{ps}^{(\I)}}p_{s,\mathrm{buf}}
-\Delta_s\right)\times\\\quad\times q_{pp}\left(q_{ps}^{(\A)}-q_{ps}^{(\I)}\right) A_0(t+1)
\end{array}
}{\left(q_{ps}^{(\A)}-q_{ps}^{(\I)}\right)
\left(1-q_{pp}(1-q_{ps}^{(\I)})A_0(t+1)\right)}
\right\rceil-1
  \end{align}
and we have defined
\begin{align}
\label{A0}
&A_0(\tau)\triangleq\frac{1-q_{pp}^{D-\tau+1}
q_{ps}^{(\I)(D-\tau+1)}
}{1-q_{pp}q_{ps}^{(\I)}},
\\
&A_1(\tau)\triangleq\frac{1-q_{pp}^{D-\tau+1}}{1-q_{pp}}.
\label{A1}
\end{align}
\end{lemma}
\begin{proof}
See App.~\ref{proofofdegenerate}.
\end{proof}

\begin{remark}
Interestingly, this is the same result derived in our work~\cite{MichelusiBIC11} for $D=2$. However, therein the result was shown to hold for general
 $q_{pp}^{(\A)}\geq q_{pp}^{(\I)}$ (not necessarily a degenerate cognitive radio network),
whereas Lemma~\ref{mainthm} holds for general $D$ but only for 
a degenerate cognitive radio network scenario.
\end{remark}

\begin{figure}
    \centering
    {
\begin{tikzpicture}[->,>=stealth',shorten >=1pt,auto,node distance=1.8cm,semithick]
\tikzstyle{every state}=[fill=black,draw=black,thick,text=white,scale=1]
\node[state]         (S10)              {$1,0,\U$};
\node[state]         (S20) [right of=S10] {$2,0,\U$};
\node[state]         (S30) [right of=S20] {$3,0,\U$};
\node[state]         (S40) [right of=S30] {$4,0,\U$};
\node[state]         (S50) [right of=S40] {$5,0,\U$};
\node[state]         (S21) [below of=S20] {$2,1,\U$};
\node[state]         (S31) [below of=S30] {$3,1,\U$};
\node[state]         (S41) [below of=S40] {$4,1,\U$};
\tikzstyle{every state}=[fill=white,draw=black,thick,text=black,scale=1]
\node[state]         (S51) [below of=S50] {$5,1,\U$};
\tikzstyle{every state}=[fill=gray,draw=black,thick,text=white,scale=1]
\node[state]         (S32) [below of=S31] {$3,2,\U$};
\tikzstyle{every state}=[fill=white,draw=black,thick,text=black,scale=1]
\node[state]         (S42) [below of=S41] {$4,2,\U$};
\node[state]         (S52) [below of=S51] {$5,2,\U$};
\node[state]         (S43) [below of=S42] {$4,3,\U$};
\node[state]         (S53) [below of=S52] {$5,3,\U$};
\node[state]         (S54) [below of=S53] {$5,4,\U$};
\tikzstyle{every state}=[fill=black,draw=black,thick,text=white,scale=1]
\node[state]         (S5K) [below of=S54] {$5,0,\K$};
\node[state]         (S4K) [left of=S5K] {$4,0,\K$};
\node[state]         (S3K) [left of=S4K] {$3,0,\K$};
\node[state]         (S2K) [left of=S3K] {$2,0,\K$};
\path (S10) edge (S20);
\path (S10) edge (S21);
\path (S20) edge (S30);
\path (S20) edge (S31);
\path (S30) edge (S40);
\path (S30) edge (S41);
\path (S40) edge (S50);
\path (S40) edge (S51);
\path (S21) edge (S31);
\path (S21) edge (S32);
\path (S31) edge (S41);
\path (S31) edge (S42);
\path (S41) edge (S51);
\path (S41) edge (S52);
\path (S32) edge (S42);
\path (S32) edge (S43);
\path (S42) edge (S52);
\path (S42) edge (S53);
\path (S43) edge (S53);
\path (S43) edge (S54);
\path (S2K) edge (S3K);
\path (S3K) edge (S4K);
\path (S4K) edge (S5K);
\path (S10) edge  [bend right] (S2K);
\coordinate (Q3) at ([yshift=-1cm]S21.west); 
\path[-] (S20) edge  [bend right] (Q3);
\path[-] (S21) edge  [bend right] (Q3);
\path (Q3) edge (S3K);
\coordinate (Q4) at ([yshift=-1cm]S32.west); 
\path[-] (S30) edge  [bend right] (Q4);
\path[-] (S31) edge  [bend right] (Q4);
\path[-] (S32) edge  [bend right] (Q4);
\path (Q4) edge (S4K);
\coordinate (Q5) at ([yshift=-1cm]S43.west); 
\path[-] (S40) edge  [bend right] (Q5);
\path[-] (S41) edge  [bend right] (Q5);
\path[-] (S42) edge  [bend right] (Q5);
\path[-] (S43) edge  [bend right] (Q5);
\path (Q5) edge (S5K);
\end{tikzpicture}
}
\caption{Illustrative example of the structure of the optimal secondary access policy for the degenerate cognitive radio network;
the SU is active in the black states, idle in the white ones,
and randomly accesses the channel in the gray state;
the arrows indicate the possible state transitions (transitions to 
state $(1,0,\U)$ are omitted).}
\label{fig:degenerate_cog_net_policy}
\end{figure}
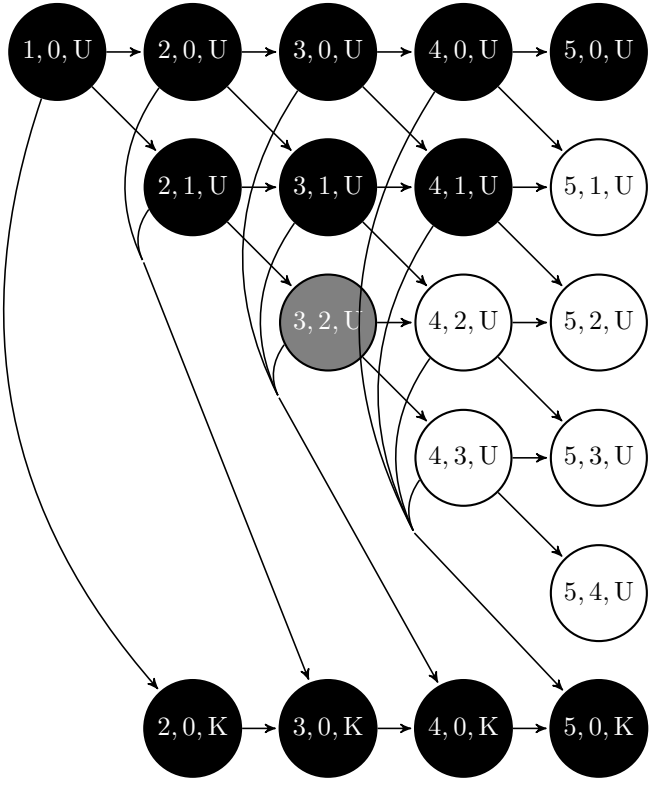

The lemma dictates that, in the degenerate cognitive radio network scenario,
the SU should restrict its channel accesses to the states
corresponding to a low primary ARQ index and small buffer occupancy at the SU receiver.
Alternatively, the larger the ARQ index or the buffer occupancy, the
smaller the incentive to access the channel.
 By doing so,
the SU maximizes the buffer occupancy in the early HARQ retransmission attempts, and invests in the future BIC recovery.
When the primary ARQ state $t$ approaches
 the deadline $D$, the SU is incentivized to idle so as to help SUrx
 to decode the PU message,
thus enabling the recovery of the failed SU transmissions
from the buffered received signals via BIC,
 before the ARQ deadline $D$ is reached and the buffer is depleted.
Moreover, when the buffer state $b$ grows,
 since $q_{ps}^{(\A)}>q_{ps}^{(\I)}$,
 the instantaneous reward accrued by staying idle
 ($(1-q_{ps}^{(\I)})bR_{s\mathrm{U}}$) 
approaches and, at some point, becomes larger than the reward accrued by transmitting
 ($T_{s\mathrm{U}}+(1-q_{ps}^{(\A)})bR_{s\mathrm{U}}$),
hence the incentive to stay idle grows.
On the other hand, if $\Delta_s$ is large, then the marginal throughput gain accrued
in the states where the PU message is known to SUrx, over
the throughput accrued in
 the states where the PU message is unknown,
is large. The SU is thus incentivized to stay idle in the initial ARQ rounds,
so as to help SUrx decode the PU message.
 Therefore, for large $\Delta_s$, the optimal policy may not
obey the structure of Lemma~\ref{mainthm}.

As a final remark, note that,
in the degenerate cognitive radio network scenario, 
the only limitation to the activity of the SU is the secondary
power expenditure $\bar{P}_{s}(\mu)$,
since the primary throughput is unaffected.
In the special case $\mathcal{P}_{s}^{(\mathrm{th})}=P_{s}$
in~(\ref{opt}),
neither the secondary power expenditure nor the primary throughput degradation
limit the activity of the SU, hence
the optimal policy solves
the unconstrained maximization problem
$\mu^{*}=
\arg\max_{\mu}\bar{T}_{s}(\mu)$,
 whose solution follows as a corollary of Lemma~\ref{mainthm}.
\begin{corollary}\label{optimaldegenerate}
 In the degenerate cognitive radio network scenario,
the solution of the unconstrained optimization problem
$\mu^{*}=
\arg\max_{\mu}\bar{T}_{s}(\mu)$ yields
 \begin{align}
\mu^*(\mathbf s)=&1,\ \forall \mathbf s\in\mathcal S_{\mathrm{K}},
\\
\mu^*(t,b,\mathrm{U})=&
  \left\{
  \begin{array}{ll}
   1 &b<\bar b_{\max}(t)\\
   0 &b\geq\bar b_{\max}(t),
  \end{array}
  \right.,\ \forall \mathbf (t,b,\U)\in\mathcal S_{\mathrm{U}},
 \end{align}
where $\bar b_{\max}(t)$ is defined in~(\ref{bmax}).
\end{corollary}

\begin{table}[t]
\small
\begin{center}
\begin{tabular}{| c | c | c |}
\hline
\multicolumn{3}{|l|}{\T\B \textbf{PU}}\\
\hline
\T\B  $R_p\simeq 2.52$&$q_{pp}^{(\mathrm{I})}\simeq 0.38$&$q_{pp}^{(\mathrm{A})}\simeq 0.68$\\
\hline
\multicolumn{3}{|l|}{\T\B \textbf{SU}, $R_{s\U }=\arg\max_{R_{s}}T_{s\U }\left(R_{s},R_p\right)$}\\
\hline
\T  $R_{s\U }=1.12$&$T_{s\U }\simeq 0.59$&\\
  $q_{ps}^{(\mathrm{I})}\simeq 0.61$&$q_{ps}^{(\mathrm{A})}\simeq 0.74$&$p_{s,\mathrm{buf}}=0.26$\\
\B $R_{s\K }\simeq 1.91$&$T_{s\K }\simeq 1.10$&\\
\hline
\multicolumn{3}{|l|}{\T\B \textbf{SU}, $R_{s\U }=R_{s\K }$}\\
\hline
\T  $R_{s\U }\simeq 1.91$&$T_{s\U }\simeq 0.40$&\\
   $q_{ps}^{(\mathrm{I})}\simeq 0.61$&$q_{ps}^{(\mathrm{A})}\simeq 0.88$&$p_{s,\mathrm{buf}}=0.37$\\
\B $R_{s\K }\simeq 1.91$&$T_{s\K }\simeq 1.10$&\\
\hline
\end{tabular}
\caption{parameters of the SU and PU, for the SNRs
$\bar\gamma_s=5$, $\bar\gamma_p=10$, $\bar\gamma_{ps}=5$, $\bar\gamma_{sp}=2$.}
\label{table}
\end{center}
\end{table}

\section{Numerical Results}
\label{sec:numres}
We consider a scenario with Rayleigh fading channels, \emph{i.e.}, the SNR $\gamma_x,\ x\in\{s,p,sp,ps\}$,
is an exponential random variable with mean $\mathbb E[\gamma_x]=\bar\gamma_x$.
We consider the following parameters, unless otherwise stated.
The average SNRs are set to $\bar\gamma_s=\bar\gamma_{ps}=5$, $\bar\gamma_p=10$,
 $\bar\gamma_{sp}=2$. The ARQ deadline is $D=5$. 
$R_{s\K }$ is chosen as
$R_{s\K }=\arg\max_{R_s}T_{s\K }(R_s)$.
The PU rate $R_p$ is chosen as the maximizer of the instantaneous PU throughput under an idle SU, \emph{i.e.},
$R_p=\arg\max_RT_{p}^{(\mathrm{I})}(R)$.
For the rate $R_{s\U }$, we evaluate the two cases 
$R_{s\U }=R_{s\U }^*$ and $R_{s\U }=R_{s\K }$,
 where $R_{s\U }^*=\arg\max_{R_s} T_{s\U }(R_s,R_{p})$.
The former maximizes the instantaneous throughput under interference from the PU,
thus neglecting the buffering capability at SUrx;
therefore, the choice $R_{s\U }=R_{s\U }^*$
reflects a pessimistic expectation of the ability
of SUrx to decode the PU message and to enable BIC.
As to the latter, from~(\ref{tau}) we have
$R_{s\U }=R_{s\K }=\arg\max_{R_s}
T_{s\U }(R_{s},R_{p})+p_{s,\mathrm{buf}}(R_{s},R_{p})R_{s\K }
$, hence $R_{s\U }=R_{s\K }$
 maximizes the sum of the instantaneous throughput and the future throughput possibly
recovered via BIC,
thus reflecting an optimistic expectation of the ability
of SUrx to decode the PU message, which enables BIC.
The PU throughput loss constraint is set to $\epsilon_{\mathrm{PU}}=0.2$,
and the constraint on the SU power is set to 
$\mathcal{P}_{s}^{(\mathrm{th})}=P_s$ (inactive).
The resulting values of the system parameters are listed in Table~\ref{table}.

We consider the following schemes: "FIC/BIC", which employs both FIC and BIC;
the optimal "FIC/BIC" policy is derived using Algorithm~\ref{algorithm} and Lemma~\ref{optpollow};
"FIC only", which does not employ the buffering mechanism~\cite{MichelusiITA11};
"no FIC/BIC", which employs neither BIC nor FIC. In this case, the SU message is
 decoded by leveraging the PU codebook structure~\cite{rateRegion};
however, possible knowledge of the PU message gained during the decoding operation
is only used in the slot where the PU message is acquired, but is neglected in the past/future
 PU retransmissions.
For "no FIC/BIC", the optimal policy consists in accessing the channel with a constant probability
 in all time-slots, independently of the underlying state, so as to attain
the  PU throughput loss constraint with equality.
"PM known" refers to an ideal scenario where SUrx perfectly knows the current PU message in advance,
and removes its interference; specifically, SUtx transmits with
rate $R_{s\K }$, thus accruing the throughput $T_{s\K }$ at each secondary access;
"PM known" thus yields an upper bound to the performance of any other policy considered.

\begin{figure}[t]
\centering  
\includegraphics[width=\linewidth,trim = 30mm 0mm 30mm 5mm,clip=true]{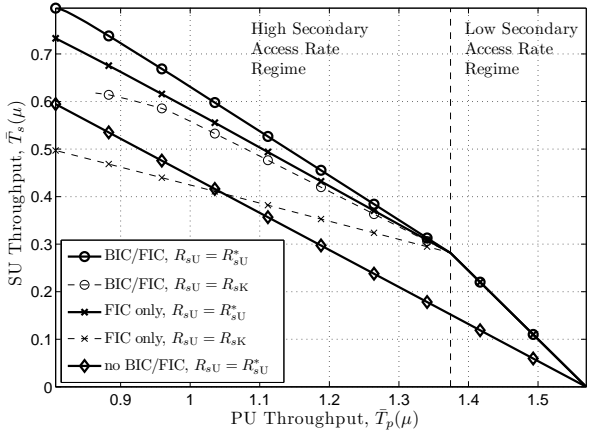}
\caption{SU throughput vs PU throughput.
 $\bar\gamma_s=\bar\gamma_{ps}=5$, $\bar\gamma_{sp}=2$, $\bar\gamma_p=10$.
The other parameters are given in Table~\ref{table}.}
\label{simresWsvsTs}
\end{figure}

In Fig.~\ref{simresWsvsTs},
 we plot the SU throughput versus the PU throughput,
obtained by varying the SU access rate constraint $\epsilon_{\mathrm{W}}$
in~(\ref{refopt2}) from $0$ to $1$.
As expected, the best performance is attained by "FIC/BIC",
 since the joint use of BIC and FIC enables IC at SUrx
 over the entire sequence of PU retransmissions. "FIC only" incurs a throughput penalty
(except in the low SU access rate regime $\bar T_p(\mu)\geq 1.37$ where, from  Lemma~\ref{optpollow},
"FIC/BIC" does not employ BIC), since
the SU transmissions which undergo outage due to severe interference from the PU are simply dropped.
"no FIC/BIC" incurs a further throughput loss, since 
possible knowledge about the PU message is not exploited to perform IC.
Concerning the choice of the transmission rates, we note that the selection
 $R_{s\U }=R_{s\U }^*$
 outperforms $R_{s\U }=R_{s\K }$ for the scenario considered.
Note that, with $R_{s\U }=R_{s\U }^*$,
the SU accrues a larger instantaneous throughput ($T_{s\U }$),
but FIC and BIC are impaired, since both the buffering probability~(\ref{omega}),
 $p_{s,\mathrm{buf}}$, and the probability that SUrx does not
successfully decode the PU message, $q_{ps}^{(\mathrm{A})}$,
diminish. Hence, in this case
 the instantaneous throughput maximization has a stronger
 impact on the performance than enabling FIC/BIC at SUrx.

\begin{figure}[t]
\centering  
\includegraphics[width=\linewidth,trim = 30mm 0mm 30mm 5mm,clip=true]{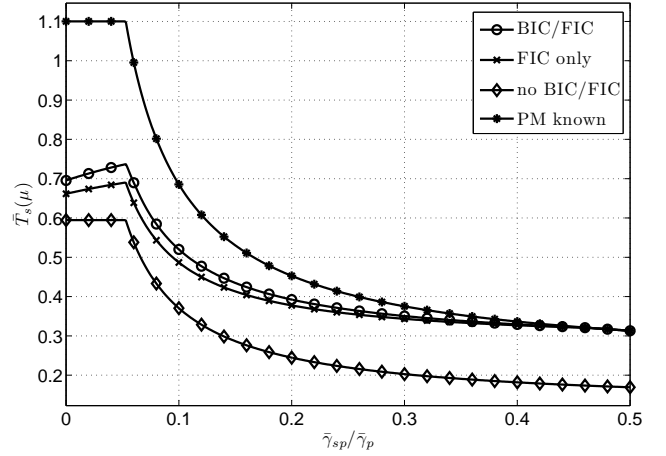}
\caption{SU throughput  vs SNR ratio $\bar\gamma_{sp}/\bar\gamma_p$.
PU throughput loss constraint $\epsilon_{\mathrm{PU}}=0.2$. 
$\bar\gamma_s=\bar\gamma_{ps}=5$, $\bar\gamma_p=10$.
$R_{s\U }=R_{s\U }^*$.}
\label{TSvsGsp}
\end{figure}

In Fig.~\ref{TSvsGsp}, we plot the SU throughput versus the
  SNR ratio $\bar\gamma_{sp}/\bar\gamma_p$, where $\bar\gamma_p=5$
and $R_{s\U }=R_{s\U }^*$. 
Note that, for $\bar\gamma_{sp}/\bar\gamma_p\leq 0.5$, the SU throughput increases.
In fact, in this regime the activity of the SU causes little
harm to the PU, and the constraint on the PU throughput loss
is inactive. The SU thus maximizes its own throughput. As $\bar\gamma_{sp}$ increases
from $0$ to $0.5\bar\gamma_{p}$,
the activity of the SU induces
more frequent primary ARQ retransmissions, hence there are more IC
opportunities available and the SU throughput augments.
On the other hand, as  $\bar\gamma_{sp}$ grows beyond $0.5\bar\gamma_p$,
the constraint on the PU throughput loss becomes active, 
 secondary accesses become more and more harmful to the PU
and take place more and more sparingly, hence the SU throughput degrades.

In Fig.~\ref{TSvsGps}, we plot the SU throughput versus the
 SNR ratio $\bar\gamma_{ps}/\bar\gamma_s$, where $\bar\gamma_s=5$
and $R_{s\U }=R_{s\U }^*$, which is a function of $\bar\gamma_{ps}$.
We notice that, when $\bar\gamma_{ps}=0$, the upper bound is achieved with equality,
since the SU operates under no interference from the PU.
The upper bound is approached also for $\bar\gamma_{ps}\gg \bar\gamma_s$,
corresponding to a strong interference regime where, with high probability,
SUrx can successfully decode the PU message,
remove its interference from the received signal, and then
 attempt to decode the SU message.
The worst performance is attained when $\bar\gamma_{ps}\simeq\bar\gamma_{s}/2$.
 In fact, the interference from the PU is neither weak enough
 to be simply treated as noise, nor
strong enough to be successfully decoded and then removed.

\begin{figure}[t]
\centering  
\includegraphics[width=\linewidth,trim = 30mm 0mm 30mm 5mm,clip=true]{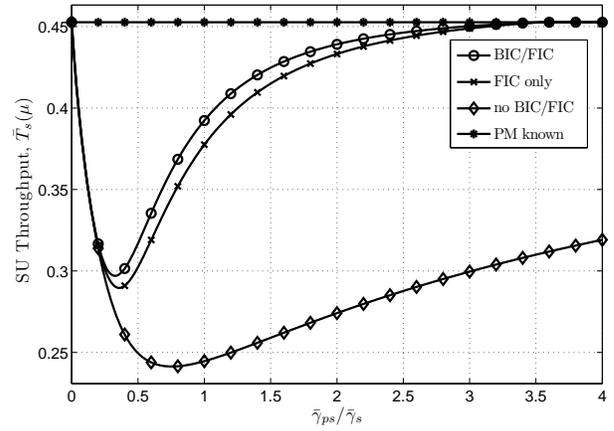}
\caption{SU throughput  vs SNR ratio $\bar\gamma_{ps}/\bar\gamma_s$.
PU throughput loss constraint $\epsilon_{\mathrm{PU}}=0.2$. $\bar\gamma_s=5$, $\bar\gamma_{sp}=2$, $\bar\gamma_p=10$.
$R_{s\U }=R_{s\U }^*$.}
\label{TSvsGps}
\end{figure}

\begin{figure}[t]
\centering  
\includegraphics[width=\linewidth,trim = 30mm 0mm 30mm 5mm,clip=true]{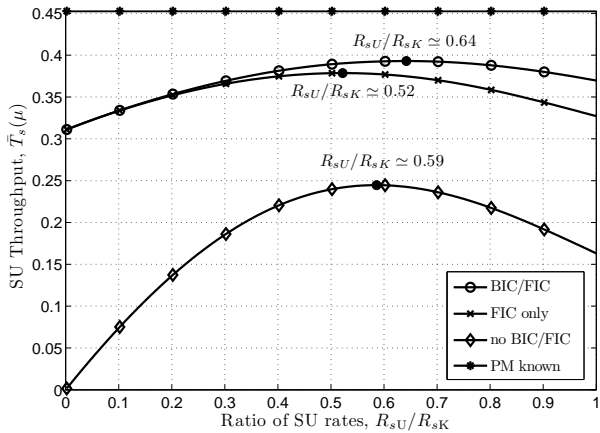}
\caption{SU throughput vs SU rate ratio $R_{s\U }/R_{s\K }$.
$R_{s\K }\simeq 1.91$ is kept fixed.
PU throughput loss constraint $\epsilon_{\mathrm{PU}}=0.2$.
 $\bar\gamma_s=5$, $\bar\gamma_{sp}=2$, $\bar\gamma_p=10$, $\bar\gamma_{ps}=5$.}
\label{varyrates}
\end{figure}

In Fig.~\ref{varyrates}, we plot the SU throughput versus the SU rate ratio $R_{s\U }/R_{s\K }$, where $R_{s\K }\simeq 1.91$ is kept fixed.
Clearly, "no FIC/BIC" attains the best performance 
for $R_{s\U }=R_{s\U }^*$, which maximizes the throughput
 $T_{s\U }(R_{s\U },R_p)$ achieved when neither FIC nor BIC are used.
On the other hand, the performance of "FIC/BIC" is maximized
 for a slightly larger value of $R_{s\U }$.
In fact, this value reflects the optimal trade-off between
 maximizing the throughput $T_{s\U }$ ($R_{s\U }\simeq 0.59 R_{s\K }$ in Fig.~\ref{varyratesprob}), 
maximizing the buffering probability, $p_{s,\mathrm{buf}}$ ($R_{s\U }\to 1$),
and minimizing the probability that SUrx does not successfully decode the PU message, $q_{ps}^{(\mathrm{A})}$ ($R_{s\U }\to 0$).
Finally, "FIC only" is optimized by 
$R_{s\U }\simeq 0.52 R_{s\K }<R_{s\U }^*$.
Since "FIC only" does not use BIC,
this value reflects the optimal trade-off between
 maximizing $T_{s\U }$ and minimizing $q_{ps}^{(\mathrm{A})}$ ($R_{s\U }\to 0$).

\begin{figure}
\centering  
\includegraphics[width=\linewidth,trim = 30mm 0mm 30mm 5mm,clip=true]{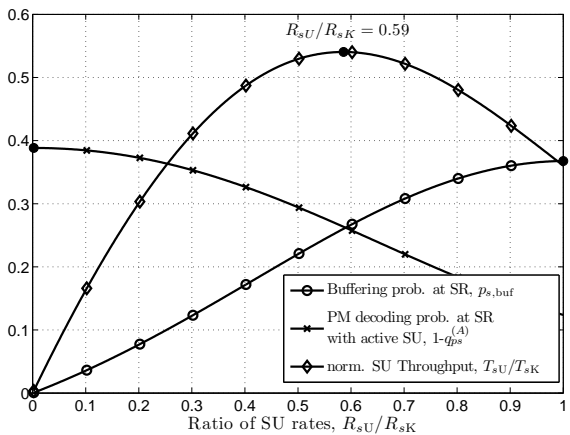}
\caption{Probabilities $p_{s,\mathrm{buf}}$, $1-q_{ps}^{(\mathrm{A})}$ and normalized SU throughput $T_{s\U }$ 
vs the SU rate ratio $R_{s\U }/R_{s\K }$.
$R_{s\K }\simeq 1.91$ is kept fixed. $\bar\gamma_s=\bar\gamma_{ps}=5$, $\bar\gamma_{sp}=2$, $\bar\gamma_p=10$.}
\label{varyratesprob}
\end{figure}
\begin{figure}
\centering  
\includegraphics[width=\linewidth,trim = 30mm 0mm 30mm 5mm,clip=true]{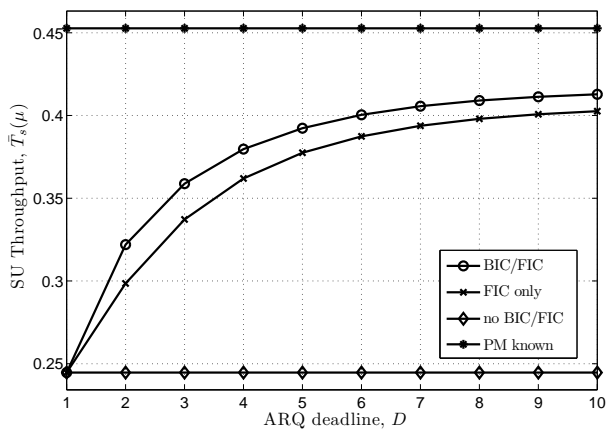}
\caption{SU throughput  vs ARQ deadline $D$.
PU throughput loss constraint $\epsilon_{\mathrm{PU}}=0.2$.
 $\bar\gamma_s=\bar\gamma_{ps}=5$, $\bar\gamma_{sp}=2$, $\bar\gamma_p=10$.
$R_{s\U }=R_{s\U }^*$.}
\label{simresvaryT}
\end{figure}

In Fig.~\ref{simresvaryT}, we plot the SU throughput versus the ARQ deadline $D$.
 We notice that, when $D=1$, all the IC mechanisms considered 
attain the same performance as "no FIC/BIC". In fact, this is a degenerate scenario where
 the PU does not employ ARQ, hence no redundancy is introduced in the primary transmission process.
Interestingly, by employing FIC or BIC, the performance improves as $D$ increases.
In fact, the larger $D$,
 the more the redundancy introduced by the primary ARQ process, hence the
more the opportunities for FIC/BIC at SUrx.
\section{Conclusion}
\label{sec:remarks}
In this work, we have investigated the idea of leveraging 
the redundancy introduced by the ARQ protocol
implemented by
a Primary User (PU)
to perform Interference Cancellation (IC) at the receiver of a Secondary User (SU) pair:
the SU receiver (SUrx), after decoding the PU message, exploits this knowledge 
to perform \emph{Forward IC} (FIC) in the following ARQ retransmissions
and \emph{Backward IC} (BIC) in the previous ARQ retransmissions,
 corresponding to SU
transmissions whose decoding failed due to severe interference from the PU.
We have employed a stochastic optimization approach to optimize the SU access strategy
which maximizes the average long-term SU throughput,
under constraints on the  average long-term PU throughput degradation
and SU power expenditure.
We have proved that the SU prioritizes its channel accesses in the states
where SUrx knows the PU message, thus enabling FIC,
and we have provided an
algorithm to optimally allocate
additional secondary access opportunities in the states 
where the PU message is unknown. 
Finally, we have shown numerically the throughput gain
of the proposed schemes.

\appendices
\section{}
\label{derivTW}
In this appendix, we compute $\bar T_s(\mu)$, $\bar W_s(\mu)$ and state properties
of $\bar W_s(\mu)$.

\begin{center}
\begin{table*}[!t]
\begin{center}
\small
\caption{Transition probabilities. $X\in\{\mathrm{A},\mathrm{I}\}$
denotes the action of the SU: active ($\mathrm{A}$) or idle ($\mathrm{I}$)}
\label{txprob}\centering
\resizebox{.8\textwidth}{!}{\hfill{}
\begin{tabular}{| c || c || c | c || c | c || c |}
\hline
\T\B\multirow{2}*{\backslashbox { From }{ To }} &
$(1,0,\U )$ & 
\multicolumn{2}{c||}{$(t+1,b,\U )$} & 
\multicolumn{2}{c||}{$(t+1,b+1,\U )$} & 
$(t+1,0,\K )$\\
\cline{2-7}\T\B
&
$X\in\{\mathrm{A},\mathrm{I}\}$ &
 $\mathrm{A}$ &
 $\mathrm{I}$ &
 $\mathrm{A}$ & 
 $\mathrm{I}$ &
$X\in\{\mathrm{A},\mathrm{I}\}$\\
\hline\T\B
$(t,b,\U )$&
 $1-q_{pp}^{(X)}$&
$q_{pp}^{(\mathrm{A})}(q_{ps}^{(\mathrm{A})}-p_{s,\mathrm{buf}})$&
$q_{pp}^{(\mathrm{I})}q_{ps}^{(\mathrm{I})}$&
$q_{pp}^{(\mathrm{A})}p_{s,\mathrm{buf}}$&
$0$&
$q_{pp}^{(X)}(1-q_{ps}^{(X)})$
 \\
\hline\T\B
$(D,b,\U )$ &
 $1$ &
 \multicolumn{2}{c||}{$0$} &
 \multicolumn{2}{c||}{$0$} &
$0$\\
\hline\T\B
$(t,0,\K )$& 
 $1-q_{pp}^{(X)}$ &
 \multicolumn{2}{c||}{$0$} &
 \multicolumn{2}{c||}{$0$}& 
$q_{pp}^{(X)}$
\\\hline\T\B
$(D,0,\K )$ &
 $1$ &
 \multicolumn{2}{c||}{$0$} &
 \multicolumn{2}{c||}{$0$} &
$0$\\\hline
\end{tabular}\hfill{}
}
\end{center}
\end{table*}
\end{center}

\begin{definition}\label{defGVD}
 We define $\mathbf G_\mu(t,b,\Phi)$, $\mathbf V_\mu(t,b,\Phi)$ and $\mathbf D_\mu(t,b,\Phi)$
as the average throughput, the average number of secondary channel accesses and 
the average number of time-slots, respectively,
 accrued starting from state $(t,b,\Phi)$ until the end of the
primary ARQ cycle under policy $\mu$ (\emph{i.e.}, until 
the recurrent state $(1,0,\U )$ is reached).
 Starting from
$\mathbf X_\mu(D+1,b,\Phi)=0,\ \forall b,\forall \Phi  \in\{\U ,\K \}$,\footnote{
We introduce the fictitious state $(D+1,b,\Phi)$ for notational convenience.
}
where $\mathbf X_\mu$ stands for $\mathbf G_\mu$, $\mathbf V_\mu$ or $\mathbf D_\mu$ 
(we write $\mathbf X\in\{\mathbf G,\mathbf V,\mathbf D\}$),
these are defined recursively as,
for $t\in\mathbb N(1,D)$, $b\in\mathbb N(0,t-1)$, 
\begin{align}\label{rec1}
\!\!\!\!\!\!\!\begin{array}{l}
 \mathbf X_\mu(t,b,\U )=x_\mu(t,b,\U )
\\\qquad+\mathrm{Pr}_\mu(t+1,b,\U |t,b,\U )\mathbf X_{\mu}(t+1,b,\U )\\
\qquad+\mathrm{Pr}_\mu(t+1,b+1,\U |t,b,\U )\mathbf X_{\mu}(t+1,b+1,\U )
\\\qquad+\mathrm{Pr}_\mu(t+1,0,\K |t,b,\U )\mathbf X_{\mu}(t+1,0,\K ),\\
\mathbf X_\mu(t,0,\K )=x_\mu(t,0,\K )
\\\qquad
+\left[q_{pp}^{(\mathrm{I})}+\mu(t,0,\K )(q_{pp}^{(\mathrm{A})}-q_{pp}^{(\mathrm{I})})\right]\mathbf X_\mu(t+1,0,\K ),
\end{array}
\end{align}
where $x_\mu(t,b,\Phi)$
 is the cost/reward accrued in state
 $(t,b,\Phi)$ and $\mathrm{Pr}_\mu(\cdot|\cdot)$ is the
 one-step transition probability, which can be derived
with the help of Table~\ref{txprob}
by taking the expectation with respect to the actions
\emph{SU idle} ($\I$, with probability $1-~\mu(t,b,\Phi)$)
and \emph{SU active} ($\A$, with probability $\mu(t,b,\Phi)$),
yielding
\begin{align}\label{onesteptxprob1}
&\mathrm{Pr}_\mu(t+1,b,\U |t,b,\U )=
\mu(t,b,\U)q_{pp}^{(A)}\left(q_{ps}^{(A)}-p_{s,\mathrm{buf}}\right)
\nonumber\\&\qquad
+(1-\mu(t,b,\U))q_{pp}^{(I)}q_{ps}^{(I)},
\\
\label{onesteptxprob2}
&\mathrm{Pr}_\mu(t+1,b+1,\U |t,b,\U )
=
\mu(t,b,\U)q_{pp}^{(A)}p_{s,\mathrm{buf}},
\\
\label{onesteptxprob3}
&
\mathrm{Pr}_\mu(t+1,0,\K |t,b,\U )
=
\mu(t,b,\U)q_{pp}^{(A)}\left(1-q_{ps}^{(A)}\right)
\nonumber\\&\qquad
+(1-\mu(t,b,\U))q_{pp}^{(I)}\left(1-q_{ps}^{(I)}\right). 
\end{align}
Namely,
if $\mathbf X=\mathbf G$ (throughput), then
$x_\mu(t,b,\Phi),\ \Phi\in\{\U,\K\}$, is the expected throughput accrued in state 
$(t,b,\Phi)$, and is given by
\begin{align}
&x_\mu(t,0,\K)=\mu(t,0,\K )T_{s\K }\triangleq g_\mu(t,0,\K ),\\
\label{bufpot}
&x_\mu(t,b,\U)=\mu(t,b,\U )T_{s\U }
\nonumber\\&
+\left[\mu(t,b,\U )(1-q_{ps}^{(\A)})+(1-\mu(t,b,\U ))(1-q_{ps}^{(\I)})\right]bR_{s\U }
\nonumber\\&\qquad \triangleq g_\mu(t,b,\U ),
\end{align}
where the second term in~(\ref{bufpot}) accounts for the successful recovery of the
$b$ SU messages from the buffered received signals via BIC,
 when the PU message is decoded by SUrx;
if $\mathbf X=\mathbf V$ (secondary access),
then $x_\mu(t,b,\Phi)$ is the SU access probability in state $(t,b,\Phi)$,
 \emph{i.e.},
\begin{align}\label{vinsta}
 x_\mu(t,b,\Phi)=\mu(t,b,\Phi)\triangleq v_\mu(t,b,\Phi);
\end{align}
finally, if $\mathbf X=\mathbf D$ (time-slots),
then
\begin{align}\label{dinsta}
 x_\mu(t,b,\Phi)=1\triangleq d_\mu(t,b,\Phi),
\end{align}
corresponding to one time-slot.
Moreover, we define, for $\mathbf X\in\{\mathbf G,\mathbf V,\mathbf D\}$,
\begin{align}\label{deriv}
 \mathbf X_{\mu}^\prime(\mathbf s)\triangleq\frac{\mathrm{d}\mathbf X_{\mu}^\prime(\mathbf s)}{\mathrm{d}\mu(\mathbf s)}.
\end{align}
\qed
\end{definition}

The number of visits
to state $(1,0,\U )$ up to time-slot $n$ is a renewal process~\cite{GallDSP}.
Each renewal interval (\emph{i.e.}, the ARQ sequence in which
the PU attempts to deliver a specific packet) has average duration
$\mathbf D_\mu(1,0,\U )$, over which the expected accrued SU throughput  
is $\mathbf G_\mu(1,0,\U )$,
and the expected number of secondary channel accesses is $\mathbf V_\mu(1,0,\U )$.
Then, the following lemma directly follows from the
strong law of large numbers for renewal-reward processes~\cite{GallDSP}.
\begin{lemma}\label{recursiveformulation}
The average long-term SU throughput and access rate are given by
$\bar T_s(\mu)=\frac{\mathbf G_\mu(1,0,\U )}{\mathbf D_\mu(1,0,\U )}$
and $\bar W_s(\mu)=\frac{\mathbf V_\mu(1,0,\U )}{\mathbf D_\mu(1,0,\U )}$,
respectively.
\qed
\end{lemma}
We have the following lemma.
\begin{lemma}
 \label{increasingWs}
We have
\begin{align}\label{p1}
 \frac{\mathrm d\bar{W}_s(\mu)}{\mathrm d\mu(\mathbf s)}\geq 0,\quad\forall \mathbf s\in\mathcal S,\ \forall\mu\in\mathcal U.
\end{align}
The inequality is strict if and only if state $\mathbf s$ is accessible from $(1,0,\U )$ under policy $\mu$, \emph{i.e.}, $\exists\ n>0:\mathrm{Pr}_\mu^{(n)}\left(\mathbf s|(1,0,\U )\right)~>~0$.
Moreover, for all $\mathbf s\in\mathcal S$ we have
\begin{align}
\label{p2}
\mathbf V_{\mu}^\prime(\mathbf s)-\mathbf D_{\mu}^\prime(\mathbf s)\bar W_s(\mu)>0.
\end{align}
\qed
\end{lemma}
\begin{proof}
If state $\mathbf s$ is not accessible from state $(1,0,\U )$ under policy $\mu$,
 then the steady state distribution satisfies 
$\pi_\mu(\mathbf s)=0$, hence $\bar W_s(\mu)$ is unaffected by $\mu(\mathbf s)$. Otherwise,
from Lemma \ref{recursiveformulation} we have that
\begin{align}
\frac{\mathrm d\bar W_s(\mu)}{\mathrm d\mu(\mathbf s)}&=
\frac{
\frac{\mathrm d\mathbf V_\mu(1,0,\U )}{\mathrm d\mu(\mathbf s)}
-\frac{\mathrm d\mathbf D_\mu(1,0,\U )}{\mathrm d\mu(\mathbf s)}\bar W_s(\mu)
}{\mathbf D_\mu(1,0,\U )}
\nonumber\\&
\propto
\mathbf V_{\mu}^\prime(\mathbf s)-\mathbf D_{\mu}^\prime(\mathbf s)\bar W_s(\mu),
\end{align}
where $\propto$ represents equality up to a positive multiplicative factor,
and the right hand side holds since,
$\forall \mathbf X\in\{\mathbf V,\mathbf D\}$ and $(t,b,\Phi)\in\mathcal S$,
$
 \frac{\mathrm d\mathbf X_\mu(1,0,\U )}{\mathrm d\mu(t,b,\Phi)}=
\mathrm{Pr}_\mu^{(t)}\left(t,b,\Phi|1,0,\U\right)
\mathbf X_{\mu}^\prime(t,b,\Phi).
$

If $\mathbf s\in\mathcal S_{\K }$, \emph{i.e.},
$\mathbf s=(t,0,\K )$, we have
\begin{align}\label{at}
 \frac{\mathrm d\bar W_s(\mu)}{\mathrm d\mu(t,0,\K )}&
\propto
\mathbf V_{\mu}^\prime(t,0,\K )-\mathbf D_{\mu}^\prime(t,0,\K )\bar W_s(\mu)
\nonumber\\&\geq
\mathbf V_{\mu}^\prime(t,0,\K )-\mathbf D_{\mu}^\prime(t,0,\K )\triangleq A_\mu(t),
\end{align}
where, from (\ref{rec1}) we have used the fact that 
$\mathbf D_{\mu}^\prime(t,0,\K )=
+(q_{pp}^{(\mathrm{A})}-q_{pp}^{(\mathrm{I})})\mathbf D_\mu(t+1,0,\K )\geq 0$
 and $\bar W_s(\mu)\leq 1$.

We now prove by induction that $A_\mu(t)>0,\forall\ t\in\mathbb N(1,T)$, 
so that (\ref{p1}) and (\ref{p2}) follow for $\mathbf s\in\mathcal S_{\K }$.
From (\ref{rec1}), for $t<D$, after algebraic manipulation
we obtain
\begin{align}\label{x}
& A_\mu(t)=1+(q_{pp}^{(\mathrm{A})}-q_{pp}^{(\mathrm{I})})
[\mathbf V_\mu(t+1,0,\K )-\mathbf D_\mu(t+1,0,\K )]\nonumber\\&
 =1-q_{pp}^{(\mathrm{A})}+\mathrm{Pr}_\mu(t+2,0,\K |t+1,0,\K )A_\mu(t+1).
\end{align}
Since $A_\mu(D)=1>0$, we obtain $A_\mu(t)>0$  by induction.

If $\mathbf s\in\mathcal S_{\U}$, \emph{i.e.},
$\mathbf s=(t,b,\U)$, we have
\begin{align}
\frac{\mathrm d\bar W_s(\mu)}{\mathrm d\mu(t,b,\U)}
\propto
\mathbf V_{\mu}^\prime(t,b,\U)-\mathbf D_{\mu}^\prime(t,b,\U)\bar W_s(\mu).
\end{align}
We prove that 
$\mathbf V_{\mu}^\prime(t,b,\U)-\mathbf D_{\mu}^\prime(t,b,\U)\bar W_s(\mu)>0$ in two steps,
so that (\ref{p1}) and (\ref{p2}) follow for $\mathbf s\in\mathcal S_{\U}$.
First,
we prove that 
$C_\mu(t,b)\triangleq\mathbf D_{\mu}^\prime(t,b,\U)\geq 0$. Then,
since $\bar W_s(\mu)\leq 1$, we obtain
\begin{align}
\frac{\mathrm d\bar W_s(\mu)}{\mathrm d\mu(t,b,0)}
&\propto
\mathbf V_{\mu}^\prime(t,b,\U)-C_\mu(t,b)\bar W_s(\mu)
\nonumber\\&
\geq
\mathbf V_{\mu}^\prime(t,b,\U)-\mathbf D_{\mu}^\prime(t,b,\U)
 \triangleq B_\mu(t,b).
\end{align}
Finally, we prove that $B_\mu(t,b)>0$.

{\bf Proof of $C_\mu(t,b)\geq 0$}: from (\ref{rec1}), for $t<D$ we have
\begin{align}
C_\mu(t,b)=&
(q_{pp}^{(\mathrm{A})}(1-q_{ps}^{(\mathrm{A})})
-q_{pp}^{(\mathrm{I})}(1-q_{ps}^{(\mathrm{I})})
)
\mathbf D_\mu(t+1,0,\K)\nonumber\\&
+(
q_{pp}^{(\mathrm{A})}(q_{ps}^{(\mathrm{A})}-p_{s,\mathrm{buf}})
-q_{pp}^{(\mathrm{I})}q_{ps}^{(\mathrm{I})}
)\mathbf D_\mu(t+1,b,\U)
\nonumber\\&
+q_{pp}^{(\mathrm{A})}p_{s,\mathrm{buf}}\mathbf D_\mu(t+1,b+1,\U).
\end{align}
Using the recursions (\ref{rec1}) and rearranging the terms, we obtain the recursive expression
\begin{align*}
&C_\mu(t,b)\!=\!\mathrm{Pr}_\mu(t\!+\!2,b\!+\!2,\U|t+1,b+1,\U)C_\mu(t+1,b+1)
\nonumber\\&
+q_{pp}^{(\mathrm{A})}\!-\!q_{pp}^{(\mathrm{I})}
+\mathrm{Pr}_\mu(t+2,b,\U|t+1,b,\U)C_\mu(t+1,b)
\\&
+\left[
(1-\mu(t+1,0,\K))q_{pp}^{(\mathrm{I})}(1-q_{ps}^{(\mathrm{I})})
\right.
\nonumber\\&
\left.+
\mu(t+1,0,\K)q_{pp}^{(\mathrm{A})}(1-q_{ps}^{(\mathrm{A})})
\right]
(q_{pp}^{(\mathrm{A})}-q_{pp}^{(\mathrm{I})})
\mathbf D_\mu(t+2,0,\K).
\nonumber
\end{align*}
Since $C_\mu(D,b)=0,\ \forall\ b\in\mathbb N(0,D-1)$,
it follows by induction on $t$ that $C_\mu(b,t)\geq 0$.

{\bf Proof of $B_\mu(t,b)>0$}: From (\ref{rec1}), for $t<D$ we obtain
the following recursive expression for $B_\mu(t,b)$, after algebraic manipulation,
\begin{align}
&B_\mu(t,b)=
1-q_{pp}^{(\mathrm{A})}+\mathrm{Pr}_\mu(t+2,b,\U|t+1,b,\U)B_\mu(t+1,b)
\nonumber\\&
+\mathrm{Pr}_\mu(t+2,b+2,\U|t+1,b+1,\U)B_\mu(t+1,b+1)
\nonumber\\&
+\left[(1-\mu(t+1,0,\K))q_{pp}^{(\mathrm{I})}(1-q_{ps}^{(\mathrm{I})})
\right.\nonumber\\&
\left.+\mu(t+1,0,\K)q_{pp}^{(\mathrm{A})}(1-q_{ps}^{(\mathrm{A})})
\right]A_\mu(t+1)
,
\end{align}
here $A_\mu(t)$ is defined in (\ref{at}).
The result follows by induction, since $B_\mu(D,b)=1>0$
and $A_\mu(t+1)>0$.
\end{proof}

\section{}
\label{rigeta}
In this appendix, we give a rigorous definition of secondary access efficiency,
 thus complementing Def.~\ref{definitioneta}.
Moreover, in Lemma~\ref{computeeta}, we derive it.
We recall that $\mathrm{Pr}_\mu^{(n)}\left(\mathbf s|\mathbf s_0\right)$ is
 the $n$-step transition probability of the chain
from $\mathbf s_0$ to $\mathbf s$.
\begin{definition}
\label{defrigeta}
Let $\tilde\mu\in\mathcal U$ be a policy such that $\exists n>0:\mathrm{Pr}_{\tilde\mu}^{(n)}\left(\mathbf s|(1,0,\U )\right)>0$,
and $\mu_\upsilon=(1-\upsilon)\mu+\upsilon\tilde\mu$, where $\upsilon~\in~(0,1]$, $\mu\in\mathcal U$.
We define the \emph{secondary access efficiency} under policy $\mu$ in state $\mathbf s\in\mathcal S$ as
\begin{align*}
 \eta_\mu\left(\mathbf s\right)=\lim_{\upsilon\to 0^+}
\left.\frac{\frac{\mathrm d\bar T_s(\mu_\upsilon)}{\mathrm d\mu_\upsilon(\mathbf s)}}{\frac{\mathrm d\bar W_s(\mu_\upsilon)}{\mathrm d\mu_\upsilon(\mathbf s)}}
\right|_{\mu_\upsilon}.
\end{align*}
\qed
\end{definition}
\begin{remark}
Notice that the condition $\exists\ n>0:\ \mathrm{Pr}_{\tilde\mu}^{(n)}\left(\mathbf s|(1,0,\U )\right)>0$ guarantees that
state $\mathbf s$ is  accessible from state $(1,0,\U )$ under policy $\mu_\upsilon$, for $\upsilon>0$.
  Under this condition, $\frac{\mathrm d\bar W_s(\mu)}{\mathrm d\mu(\mathbf s)}>0$ (Lemma~\ref{increasingWs} in App. \ref{derivTW}),
hence the fraction within the limit is well defined for $\upsilon>0$ and in the limit $\upsilon\to 0^+$.
One such policy $\tilde\mu$ is $\tilde\mu(\mathbf s)=0.5,\forall\mathbf s\in\mathcal S$. 
\qed
\end{remark}

Using Lemma~\ref{recursiveformulation} and  Def.~\ref{defGVD}
in App.~\ref{derivTW}
 and Def.~\ref{defrigeta},
 $\eta_\mu\left(\mathbf s\right)$ can be derived according
to the following lemma. 
\begin{lemma}\label{computeeta}
We have
$
 \eta_\mu\left(\mathbf s\right)=
\frac{
\mathbf G_{\mu}^\prime(\mathbf s)
-\mathbf D_{\mu}^\prime(\mathbf s)\bar T_s(\mu)
}
{
\mathbf V_{\mu}^\prime(\mathbf s)
-\mathbf D_{\mu}^\prime(\mathbf s)\bar W_s(\mu)
}. 
$
\qed
\end{lemma}
\begin{remark}
This is well defined, since 
$ \mathbf V_{\mu}^\prime(\mathbf s)-\mathbf D_{\mu}^\prime(\mathbf s)\bar W_s(\mu)>0$
from Lemma~\ref{increasingWs} in App. \ref{derivTW}.
\qed
\end{remark}

 \section{}
\label{proofofalgo}
 \begin{proof}[Proof of Theorem~\ref{algooptimal}]
In the first part of the theorem, we prove that,
 by initializing Algorithm~\ref{algorithm}
with the idle policy $\mu^{(0)}$, $\mu^{(0)}(\mathbf s)=0,\ \forall\mathbf s\in\mathcal S$,
and with the set of idle states
$\mathcal S_{\mathrm{idle}}^{(0)}\equiv\mathcal S$,
we obtain an optimal policy.
In the second part of the proof, we prove the optimality of the specific initialization of Algorithm~\ref{algorithm}
for the high SU access rate regime.

 Let $\tilde\mu$ be a policy under which all states
 $\mathbf s\in\mathcal S$ are  accessible from state $(1,0,\U )$,
\emph{i.e.},  $\exists\ n>0:\ \mathrm{Pr}_{\tilde\mu}^{(n)}\left(\mathbf s|(1,0,\U )\right)>0$.
One  such policy is $\tilde\mu(\mathbf s)=0.5,\ \forall\mathbf s\in\mathcal S$.
Consider a modified Markov Decision Process, parameterized by $\upsilon\in (0,1)$,
obtained by applying the policy $(1-\upsilon)\mu+\upsilon\tilde\mu$ to the original system, where $\mu\in\mathcal U$.
Since $\mu,\tilde\mu\in\mathcal U$ and $\upsilon\in (0,1)$, it follows that 
$(1-\upsilon)\mu+\upsilon\tilde\mu\in\mathcal U$.
We define $\bar{T}_s(\mu,\upsilon)\triangleq\bar{T}_s((1-\upsilon)\mu+\upsilon\tilde\mu)$ and
$\bar{W}_s(\mu,\upsilon)\triangleq\bar{W}_s((1-\upsilon)\mu+\upsilon\tilde\mu)$,
 and we study the problem
\begin{align}
\label{refopt3}
\mu^{*(\upsilon)}=&
 {\arg\max}_{\mu\in\mathcal U}\bar{T}_s(\mu,\upsilon)\ \mathrm{s.t.\  }\bar{W}_s(\mu,\upsilon)\leq\epsilon_{\mathrm{W}},
\end{align}
where the parameter $\upsilon$ is
small enough to guarantee a feasible problem, \emph{i.e.}, $\exists\ \mu\in\mathcal U:\bar{W}_s(\mu,\upsilon)\leq\epsilon_{\mathrm{W}}$.
(\ref{refopt2}) is obtained in the limit $\upsilon\to 0^+$.
Notice that, $\forall\ \mu\in\mathcal U$,
 under policy $(1-\upsilon)\mu+\upsilon\tilde\mu$,
 all the states $\mathbf s\in\mathcal S$
are accessible from state $(1,0,\U )$, and the Markov chain is irreducible.
Hence, from Lemma~\ref{increasingWs} in App. \ref{derivTW},
$\bar W_s(\mu,\upsilon)$ is a strictly increasing function of $\mu(\mathbf s),\ \forall\mathbf s\in\mathcal S$.
 This is an important assumption in the following proof.

Let $\mathcal D\subset\mathcal U$ be the
set of all the deterministic policies, and 
$\mathcal G_{\upsilon}=\left\{\left(\bar W_s(\mu,\upsilon),
\bar T_s(\mu,\upsilon)\right),\mu\in\mathcal D\right\}$.
With the help of Fig.~\ref{fig:conv}, for any $\mu\in\mathcal U$, we have that 
$\left(\bar W_s(\mu,\upsilon),\bar T_s(\mu,\upsilon)\right)\in\mathrm{conv}(\mathcal G_\upsilon)$, where $\mathrm{conv}(\mathcal G_\upsilon)$ is the convex hull of the set $\mathcal G_\upsilon$.
In particular, for the optimal policy we have
$\left(\bar W_s(\mu^{*(\upsilon)},\upsilon),\bar T_s(\mu^{*(\upsilon)},\upsilon)\right)\in\mathrm{bd}(\mathcal G_\upsilon)$,
 where $\mathrm{bd}(\mathcal G_\upsilon)$ denotes the boundary of $\mathrm{conv}(\mathcal G_\upsilon)$.

\begin{figure}[t]
\centering  
\includegraphics[width=\linewidth,trim = 30mm 5mm 30mm 0mm,clip=true]{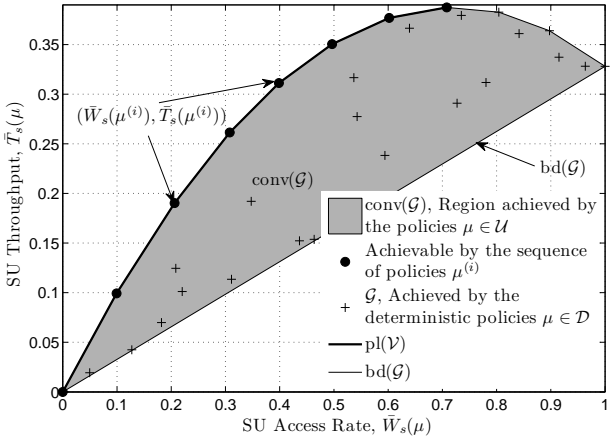}
\caption{Geometric interpretation of problem~(\ref{refopt3})}
\label{fig:conv}
\end{figure}

Algorithm~\ref{algorithm} determines the sequence of vertices of the polyline
 $\mathrm{bd}(\mathcal G_\upsilon)$ in the limit $\upsilon\to 0^+$ (bold line in Fig.~\ref{fig:conv}).
For $\upsilon>0$,
starting from the leftmost vertex of $\mathrm{bd}(\mathcal G_\upsilon)$, achieved
 by the \emph{idle} policy $\mu^{(0)}(\mathbf s)=0,\ \forall\mathbf s\in\mathcal S$
(this follows from the fact that  $\bar W_s(\mu,\upsilon)$
is a strictly increasing function of $\mu(\mathbf s)$, hence it is minimized
by the idle policy),
the algorithm determines iteratively the next vertex of $\mathrm{bd}(\mathcal G_\upsilon)$ as the maximizer of the slope
\begin{align}\label{maxim}
\mu^{(i+1)}=\!\!\!\!\!\!\!\underset{\mu\in\mathcal D:\bar W_s(\mu,\upsilon)>\bar W_s(\mu^{(i)},\upsilon)}{\arg\max} 
\frac{\bar T_s(\mu,\upsilon)-\bar T_s(\mu^{(i)},\upsilon)}{\bar W_s(\mu,\upsilon)-\bar W_s(\mu^{(i)},\upsilon)}.
\end{align}
Since~(\ref{refopt2}) has one constraint, 
the optimal policy $\mu^{*(\upsilon)}$ is randomized in one state~\cite{Ross1989},
and hence each segment on the boundary $\mathrm{bd}(\mathcal G_\upsilon)$ between pairs 
$(\bar W_s(\mu^{(i)},\upsilon),\bar T_s(\mu^{(i)},\upsilon))$
 achievable with deterministic policies is attained by a policy
  that is randomized in only one state.
It follows that $\mu^{(i)}$ and $\mu^{(i+1)}$ differ in only one state.
Moreover, in~(\ref{maxim}) the maximization is over $\mu\in\mathcal D$ such that $\bar W_s(\mu,\upsilon)>\bar W_s(\mu^{(i)},\upsilon)$,
\emph{i.e.}, since $\bar W_s(\mu,\upsilon)$ is a strictly increasing function of $\mu(\mathbf s)$
and $\mu^{(i+1)}$
and $\mu^{(i)}$ differ in only one position,
$\mu^{(i+1)}$ is obtained from $\mu^{(i)}$ by allocating one more secondary access
to a state which is idle under $\mu^{(i)}$.
In~(\ref{maxim}), the maximization is thus over
$\left\{\mu^{(i)}+\delta_{\mathbf s}:\mathbf s\in\mathcal S_{\mathrm{idle}}^{(i)}\right\}$,
 and, after algebraic manipulation, $ \mu^{(i+1)}$ in~(\ref{maxim}) maximizes
\begin{align*}
\max_{\mathbf s\in\mathcal S_{\mathrm{idle}}^{(i)}}
\!\!\frac{\bar T_s(\mu^{(i)}\!\!+\delta_{\mathbf s},\upsilon)-\bar T_s(\mu^{(i)}\!,\upsilon)}
{\bar W_s(\mu^{(i)}\!\!+\delta_{\mathbf s},\upsilon)-\bar W_s(\mu^{(i)}\!,\upsilon)}
\!=\!\!\max_{\mathbf s\in\mathcal S_{\mathrm{idle}}^{(i)}}\!\!\eta_{(1-\upsilon)\mu^{(i)}+\upsilon\tilde\mu}(\mathbf s).
\end{align*}
\emph{Stage $i$} of the algorithm is thus proved.
If $\eta_{(1-\upsilon)\mu^{(i)}+\upsilon\tilde\mu}(\mathbf s)\leq 0$,
we have $\bar W_s\left(\mu^{(i)}+\delta_{\mathbf s},\upsilon\right)>\bar W_s\left(\mu^{(i)},\upsilon\right)$ and 
$\bar T_s\left(\mu^{(i)}+\delta_{\mathbf s},\upsilon\right)\leq\bar T_s\left(\mu^{(i)},\upsilon\right)$.
If this condition holds $\forall\ \mathbf s\in\mathcal S_{\mathrm{idle}}^{(i)}$,
 any next vertex of the polyline $\mathrm{bd}(\mathcal G_\upsilon)$ yields
a decrease of the SU throughput and a larger SU access rate, hence a sub-optimal set of policies,
and the algorithm stops.

By construction, the algorithm returns a sequence of policies $(\mu^{(i)},i\in\mathbb N(0,N-1))$, characterized by strictly increasing values of
the SU throughput and of the SU access rate.
The optimal policy belongs to the polyline with vertices 
$\mathcal V_\upsilon\equiv\{(\bar W_s(\mu^{(i)},\upsilon),\bar T_s(\mu^{(i)},\upsilon)),i\in\mathbb N(0,N-1)\}$,
denoted by
$\mathrm{pl}(\mathcal V_\upsilon)$ in Fig.~\ref{fig:conv}.
Then,~(\ref{refopt2}) becomes equivalent to
$T_s^{*(\upsilon)}=\underset{(W_s,T_s)\in\mathcal V_\upsilon}{\max}T_s \mathrm{\ s.t.\  }W_s\leq \epsilon_{\mathrm{W}},
$
whose solution is given in the last step of Algorithm~\ref{algorithm}.
The result finally follows for $\upsilon\to 0^+$.

To conclude, we prove the initialization of Algorithm~\ref{algorithm} for the high SU access rate. 
Let $(\mu^{(0)},\dots,\mu^{(N-1)})$ and $(\mathbf s^{(0)},\dots,\mathbf s^{(N-1)})$ be the sequence of deterministic policies
and of states returned by Algorithm \ref{algorithm},
obtained by initializing the algorithm as in the first part of the proof.
Let $\mathcal D_0\equiv\left\{\mu\in\mathcal D:\mu(t,0,0)=0\ \forall\ t\in\mathbb N(1,T)\right\}$,
$\tilde{\mathcal D}_0\equiv\left\{\mu\in\mathcal D_0:\mu(\mathbf s)=1,\ \forall\mathbf s\in\mathcal S_{\K}\right\}$,
and $N_0\triangleq\max\{i\in\{0,\dots,N-1\}:\bar W_s(\mu^{(i)})<\epsilon_{\mathrm{th}}\}$.
We prove that $\mu^{(N_0+1)}\in\tilde{\mathcal D}_0$, \emph{i.e.},
$\mu^{(N_0+1)}(\mathbf s)=1,\ \forall \mathbf s\in\mathcal S_{\K}$.
From the definition of $\tilde{\mathcal D}_0$ and the construction of the algorithm,
it follows that, for $i>N_0$,
$\mu^{(i)}(\mathbf s)=1,\ \forall\ \mathbf s\in\mathcal S_{\K}$.
Moreover, from Lemma \ref{D0}, $\bar W_s(\mu^{(N_0+1)})=\epsilon_{\mathrm{th}}$.
Hence, for the high SU access rate $\epsilon>\epsilon_{th}$,
 the optimal policy $\mu^{*}$ obeys
$\mu^{*}(\mathbf s)=1,\forall\mathbf s\in\mathcal S_{\K}$.
Then, letting $\mathcal U_1\equiv\{\mu\in\mathcal U:\mu(\mathbf s)=1,\forall\mathbf s\in\mathcal S_{\K}\}$,
 the optimization problem (\ref{refopt2}) can be restricted to the set of randomized policies
$\mu\in\mathcal U_1\subset\mathcal U$ when $\epsilon>\epsilon_{th}$.
Equivalently,
secondary accesses taking place in $\mathcal S_{\U }$ can be obtained
by initializing the algorithm with
$\mu^{(0)}(\mathbf s)=0,\ \mathbf s\in\mathcal S_{\U}$,
$\mu^{(0)}(\mathbf s)=1,\ \mathbf s\in\mathcal S_{\K}$,
$S_{\mathrm{idle}}^{(0)}\equiv\mathcal S_{\U}$.
The initialization of Algorithm~\ref{algorithm} is thus proved.

{\bf Proof of $\mu^{(N_0+1)}\in\tilde{\mathcal D}_0$}:
We prove by induction that $\mu^{(i)}\in\mathcal D_0\setminus\tilde{\mathcal D}_0,\ \forall i\leq N_0$ and
 $\mu^{(N_0+1)}\in\tilde{\mathcal D}_0$.
Assume that, for some $i\geq 0$, $\mu^{(j)}\in\mathcal D_0\setminus\tilde{\mathcal D}_0,\ \forall j\leq i$.
From Lemma \ref{D0}, it follows that $N_0\geq i$.
This clearly holds for $i=0$.
We show that this implies that
 either $\mu^{(i+1)}\in\mathcal D_0\setminus\tilde{\mathcal D}_0$, hence $N_0>i$,
thus proving the induction step, or
 $\mu^{(i+1)}\in\tilde{\mathcal D}_0$, hence $N_0=i$,
thus proving the property.
The result follows since $N_0\leq 1+|\mathcal S|<\infty$ (\emph{i.e.}, $i=N_0$ is reached within a finite number of steps).

From Lemma \ref{useful}, $\eta_{\mu^{(i)}}(\mathbf s)=T_{s\K}>0,
\forall\mathbf s\in\mathcal S_{\K}\cap\mathcal S_{\mathrm{idle}}^{(i)}$
and $\eta_{\mu^{(i)}}(t,0,\U)<T_{s\K},\forall t\in\mathbb N(1,D)$,
hence, from the main iteration stage of the algorithm it follows that
$\mu^{(i+1)}\in\mathcal D_0$.
In particular, if 
$\mu^{(i+1)}\in\mathcal D_0\setminus\tilde{\mathcal D}_0$,
then $N_0>i$ from Lemma \ref{D0}.
On the other hand, if $\mu^{(i+1)}\in\tilde{\mathcal D}_0$,
then, from Lemma \ref{D0}, $N_0=i$. The property is thus proved.
\end{proof}

\begin{lemma}\label{D0}
$ $\\$\bar W_s(\mu)<\epsilon_{\mathrm{th}},\forall\mu\in\mathcal D_0\setminus\tilde{\mathcal D}_0$
and
$\bar W_s(\mu)=\epsilon_{\mathrm{th}},\forall\mu\in\tilde{\mathcal D}_0$.\qed
\end{lemma}
\begin{proof}
 Let $\mu\in\tilde{\mathcal D}_0$. Since the states $(t,b,\U)$ with $b>0$ are not 
accessible from $(1,0,\U)$ under $\mu$, 
the transmission probability $\mu(t,b,\U)$, $b>0$, does not affect $\bar W_s(\mu)$.
Then, from Def. \ref{threshold}, we have $\bar W_s(\mu)=\epsilon_{\mathrm{th}}$.

Let $\mu\in\mathcal D\setminus\tilde{\mathcal D}_0$.
Letting $\mathcal S_\mu=\{\mathbf s\in\mathcal S_{\K}:\mu(\mathbf s)=0\}$,
we have that $\mu+\sum_{\mathbf s\in\mathcal S_\mu}\delta_{\mathbf s}\in\tilde{\mathcal D}_0$.
Finally, since every $\mathbf s\in\mathcal S_\mu$ is accessible from $(1,0,\U)$ under $\mu$,
and $\mathcal S_\mu$ is non-empty,
from Lemma~\ref{increasingWs} in App. \ref{derivTW}
and the previous case,
 it follows that
$\bar W_s(\mu)<\bar W_s(\mu+\sum_{\mathbf s\in\mathcal S_\mu}\delta_{\mathbf s})=\epsilon_{\mathrm{th}}$.
\end{proof}

\begin{lemma}\label{useful}
Let $\mu\in\mathcal U$ such that 
$\mu(t,0,\U)=0\ \forall\ t\in\mathbb N(1,D)$.
 Then, $\eta_{\mu}(t,0,\U)<T_{s\K}$ and
$\eta_{\mu}(t,0,\K)=T_{s\K},\ \forall t$.\qed
\end{lemma}
\begin{proof}
Let $\mu\in\mathcal U$ such that 
$\mu(t,0,\U)=0\ \forall\ t\in\mathbb N(1,D)$.
It follows that
the states $(t,b,\U)$ with $b>0$ are not accessible, hence their steady state probability satisfies
$\pi_\mu(t,b,\U)=0,\ \forall\ t,\ \forall\ b>0$.
It is then straightforward to show, by using the recursion (\ref{rec1}), that
$\mathbf G_\mu(t,0,\U)=T_{s\K}\mathbf V_\mu(t,0,\K)$,
 $\mathbf G_\mu(t,0,\K)=T_{s\K}\mathbf V_\mu(t,0,\K)$ and
$\bar T_{s}(\mu)=T_{s\K}\bar W_s(\mu)$.
Then, using these expressions, the recursion (\ref{rec1}) and 
Lemma \ref{computeeta},
we obtain $\eta_{\mu}(t,0,\K)=T_{s\K}$ and
\begin{align}
&\eta_{\mu}(t,0,\U)=
T_{s\K}
-
\frac{
T_{s\K}\mathbf V_{\mu}^\prime(t,0,\U)-\mathbf G_{\mu}^\prime(t,0,\U)
}
{
\mathbf V_{\mu}^\prime(t,0,\U)
-\mathbf D_{\mu}^\prime(t,0,\U)\bar W_s(\mu)
}.
\end{align}
We now prove that $\eta_{\mu}(t,0,\U)<T_{s\K}$, which proves the lemma.
Equivalently,
using Lemma~\ref{increasingWs} in App. \ref{derivTW} and (\ref{rec1}),
we prove that
\begin{align}
\label{mt1}
&T_{s\K}\mathbf V_{\mu}^\prime(t,0,\U)-\mathbf G_{\mu}^\prime(t,0,\U)=
(T_{s\K}-T_{s\U})
\\&
+q_{pp}^{(\mathrm{A})}p_{s,\mathrm{buf}}
[T_{s\K}\mathbf V_{\mu}(t,1,\U)-\mathbf G_{\mu}(t,1,\U)]
>0.
\nonumber
\end{align}
Letting
\begin{align}
&M_\mu(t,b)=
b(T_{s\K}-T_{s\U})\\\nonumber&+q_{pp}^{(\mathrm{A})}p_{s,\mathrm{buf}}
[T_{s\K}\mathbf V_{\mu}(t,b,\U)-\mathbf G_{\mu}(t,b,\U)]>0,\ \forall\ t,b\geq 1,
\end{align}
(\ref{mt1}) is equivalent to $M_\mu(t,1)>0$.
We now prove by induction that $M_\mu(t,b)>0,\ \forall\ t,b\geq 1$, yielding (\ref{mt1})
as a special case when $b=1$.
For $t=D+1$ we have $M_\mu(D+1,b)=b(T_{s\K}-T_{s\U})>0$, since $T_{s\K}>T_{s\U}$ and $b\geq 1$.
Now, let $t\leq D$ and  assume  $M_\mu(t+1,b)>0$.
Using (\ref{rec1}), after algebraic manipulation we obtain
\begin{align}
&M_\mu(t,b)=b(T_{s\K}-T_{s\U})
+q_{pp}^{(\mathrm{A})}p_{s,\mathrm{buf}}\mu(t,b,\U)(T_{s\K}-T_{s\U})
\nonumber\\&
-q_{pp}^{(\mathrm{A})}p_{s,\mathrm{buf}}\left[
1
-\mu(t,b,\U)q_{ps}^{(\mathrm{A})}-(1-\mu(t,b,\U))q_{ps}^{(\mathrm{I})}\right]bR_{s\U}
\nonumber\\&
+\mathrm{Pr}_\mu(t+1,b,\U|t,b,\U)[M_\mu(t+1,b)-b(T_{s\K}-T_{s\U})]
\nonumber\\&
+\mathrm{Pr}_\mu(t+1,b+1,\U|t,b,\U)M_\mu(t+1,b)
\nonumber\\&
-\mathrm{Pr}_\mu(t+1,b+1,\U|t,b,\U)(b+1)(T_{s\K}-T_{s\U}).
\end{align}
Finally, since $M_\mu(t+1,b)>0$ by the induction hypothesis,
using inequality (\ref{tau})
 we obtain
\begin{align}
&M_\mu(t,b)
>
p_{s,\mathrm{buf}}bR_{s\U}\left(
1-q_{pp}^{(\mathrm{A})}
\right)
\nonumber\\&
+p_{s,\mathrm{buf}}bR_{s\U}
(1-\mu(t,b,\U))q_{ps}^{(\mathrm{I})}(q_{pp}^{(\mathrm{A})}-q_{pp}^{(\mathrm{I})})
>0,
\end{align}
which proves the induction step.
The lemma is proved.
\end{proof}

 \section{}
\label{proofofdegenerate}

\begin{proof}[Proof of Lemma~\ref{mainthm}]
Let $\mathcal D\subset\mathcal U$ be the
set of all the deterministic (non-randomized) policies.
 Let
 \begin{align}\label{strucxxx}
\begin{array}{ll}
\tilde{\mathcal D}\equiv&\left\{
  \mu\in\mathcal D:
  \mu(t,b,\mathrm{U})=1,\forall t,b<b(t);\right.\\&\left.\ 
  \mu(t,b,\mathrm{U})=0,\forall t,b\geq b(t);\ 
  \mu(\mathbf s)=1,\mathbf s\in\mathcal S_{\mathrm{K}};\right.\\&\left.
\    \exists\ b(\cdot):b(t+1)\leq b(t)\ \forall t
  \right\}.\nonumber
\end{array}
 \end{align}
By inspection, we have that
the sequences of policies~(\ref{struc}) are such that
 $\mu^{(i)}\in\tilde{\mathcal D},\ \forall i\in\mathbb N(0,N-1)$. 
Therefore, the first part of the lemma states that $\mu^{(i)}\in\tilde{\mathcal D},\ \forall i\in\mathbb N(0,N-1)$. 
We prove this property by induction. 
Namely, we show that $\mu^{(i)}\in\tilde{\mathcal D}\Rightarrow\mu^{(i+1)}\in\tilde{\mathcal D}$.
Then, since $\mu^{(0)}\in\tilde{\mathcal D}$ (initialization of Algorithm~\ref{algorithm})
 it follows that $\mu^{(i)}\in\tilde{\mathcal D},\ \forall i$.
Let $\mu^{(i)}\in\tilde{\mathcal D}$, \emph{i.e.}, $\mu^{(i)}$ is given by~(\ref{struc})
for some $b^{(i)}(t)$ non-increasing in $t$.
The set of idle states is then given by
\begin{align}
 \mathcal S_{\mathrm{idle}}^{(i)}\equiv
\left\{(t,b,\mathrm{U})\in\mathcal S_{\mathrm{U}}:t\in\mathbb N(1,D), b\geq b^{(i)}(t)\right\}.
\end{align}
We then prove that, under the hypotheses of the lemma, $\eta_{\mu^{(i)}}(t,b,\mathrm{U})>\eta_{\mu^{(i)}}(t,b+1,\mathrm{U})$
and $\eta_{\mu^{(i)}}(t,b,\mathrm{U})>\eta_{\mu}(t+1,b,\mathrm{U}),\ \forall (t,b,\mathrm{U})\in\mathcal S_{\mathrm{idle}}^{(i)}$.
It follows that 
the SU access efficiency is maximized by the state
in the idle set $\mathcal S_{\mathrm{idle}}^{(i)}$
with the lowest value of the primary ARQ state $t$, among
the states with the same buffer occupancy $b$,
and with the fewest number of 
buffered received signals $b$, among
the states with the same primary ARQ state $t$.
Therefore, in the main iteration stage of the algorithm,
the SU access efficiency is maximized by
$\mathbf s^{(i)}=\arg\max_{\mathbf s\in\mathcal S_{\mathrm{idle}}^{(i)}}\eta_{\mu^{(i)}}(\mathbf s)$,
where $\mathbf s^{(i)}=(t,b,\U)$ is such that 
$\tau\geq t$, $\beta\geq b$,
 $\forall\ (\tau,\beta,\U)\in\mathcal S_{\mathrm{idle}}^{(i)}$.
By inspection, we have that $\mu^{(i+1)}=\mu^{(i)}+\delta_{\mathbf s^{(i)}}\in\tilde{\mathcal D}$, hence the induction step is proved.

We thus need to prove the induction step, \emph{i.e.}, letting $\mu^{(i)}\in\tilde{\mathcal D}$,
we show that
\begin{align}
&\eta_{\mu^{(i)}}(t,b,\mathrm{U})>\eta_{\mu^{(i)}}(t,b+1,\mathrm{U}),\ \forall (t,b,\mathrm{U})\in\mathcal S_{\mathrm{idle}}^{(i)},
\nonumber\\
&\eta_{\mu^{(i)}}(t,b,\mathrm{U})>\eta_{\mu}(t+1,b,\mathrm{U}),\ \forall (t,b,\mathrm{U})\in\mathcal S_{\mathrm{idle}}^{(i)}.
\end{align}
To this end, note that, in the degenerate cognitive radio network scenario,
the primary ARQ process is not affected by the SU access scheme,
hence, using the notation in App.~\ref{derivTW}, $\mathbf D_{\mu^{(i)}}^\prime(t,b,\U)=0$.
By the definition of SU access efficiency~(\ref{computeeta}),
we thus obtain
\begin{align}\label{etaspecial}
 \eta_{\mu^{(i)}}\left(t,b,\U \right)=\frac{\mathbf G_{\mu^{(i)}}^\prime(t,b,\U )}
{\mathbf V_{\mu^{(i)}}^\prime(t,b,\U )},
\end{align}
where, using (\ref{rec1}), (\ref{onesteptxprob1}-\ref{onesteptxprob3}),
(\ref{bufpot}) and (\ref{vinsta}),
\begin{align}\label{Gprimex}
& \mathbf G_{\mu^{(i)}}^\prime(t,b,\U )=
T_{s\U }+\left(q_{ps}^{(\I)}-q_{ps}^{(\A)}\right)bR_{s\U }
\\&\nonumber
\quad+q_{pp}(q_{ps}^{(\A)}-p_{s,\mathrm{buf}}-q_{ps}^{(\I)})\mathbf G_{\mu^{(i)}}(t+1,b,\U )\\
&\quad+q_{pp}p_{s,\mathrm{buf}}\mathbf G_{\mu^{(i)}}(t+1,b+1,\U )
\nonumber\\
&\quad+q_{pp}(q_{ps}^{(\I)}-q_{ps}^{(\A)})\mathbf G_{\mu^{(i)}}(t+1,0,\K ),\nonumber\\ 
\label{Vprimex}
&\mathbf V_{\mu^{(i)}}^\prime(t,b,\U )=1
+q_{pp}(q_{ps}^{(\A)}-p_{s,\mathrm{buf}}-q_{ps}^{(\I)})\mathbf V_{\mu^{(i)}}(t+1,b,\U )\nonumber\\
&\quad+q_{pp}p_{s,\mathrm{buf}}\mathbf V_{\mu^{(i)}}(t+1,b+1,\U )
\nonumber\\&\quad+q_{pp}(q_{ps}^{(\I)}-q_{ps}^{(\A)})\mathbf V_{\mu^{(i)}}(t+1,0,\K ).
\end{align}
Using the fact that $\mu^{(i)}(\tau,\beta,\U)=0,\ \forall \tau\geq t,\beta\geq b$,
it can be proved that 
\begin{align}\label{v1}
&\mathbf V_{\mu^{(i)}}(\tau,\beta,\U)=A_1(\tau)-A_0(\tau),\\
\label{v2}
&\mathbf G_{\mu^{(i)}}(\tau,\beta,\U)=(1-q_{ps}^{(\I)})
\beta R_{s\U}A_0(\tau)\\&\nonumber\qquad+T_{s\K}(A_1(\tau)-A_0(\tau)),\\
\label{v3}
&\mathbf V_{\mu^{(i)}}(\tau,0,\K)=A_1(\tau),\\
\label{v4}
&\mathbf G_{\mu^{(i)}}(\tau,0,\K)=T_{s\K}A_1(\tau),
\end{align}
where $A_0(\cdot)$ and $A_1(\cdot)$ are defined in
(\ref{A0}) and (\ref{A1}), respectively.
The expressions~(\ref{v1}-\ref{v4}) can be easily verified by induction,
 starting from $\tau=D+1$ backward. 
In fact, for $\tau=D+1$, we have
$A_0(D+1)=A_1(D+1)=0$,
hence we obtain $
\mathbf V_{\mu^{(i)}}(D+1,\beta,\U)=
\mathbf G_{\mu^{(i)}}(D+1,\beta,\U)=
\mathbf V_{\mu^{(i)}}(D+1,0,\K)=
\mathbf G_{\mu^{(i)}}(D+1,0,\K)=0$,
which is consistent with Def.~\ref{defGVD}.
The induction step can be proved
by inspection, using the recursive expression~(\ref{rec1})
and the fact that $\mu(\tau,\beta,\U)=0,\ \forall\tau\geq t,\beta\geq b$. 
 Substituting the expressions~(\ref{v1}-\ref{v4}) in~(\ref{Gprimex})
and~(\ref{Vprimex}), we obtain
\begin{align}\label{Gprime}
&\mathbf G_{\mu^{(i)}}^\prime(t,b,\U )=
T_{s\U }+q_{pp}p_{s,\mathrm{buf}}(1-q_{ps}^{(\I)})R_{s\U}A_0(t+1)\nonumber
\\\nonumber &
 \quad+\left(q_{ps}^{(\I)}-q_{ps}^{(\A)}\right)bR_{s\U }\left[1-q_{pp}(1-q_{ps}^{(\I)})A_0(t+1)\right]
\\&
\quad 
 +q_{pp}(q_{ps}^{(\I)}-q_{ps}^{(\A)})T_{s\K}A_0(t+1),
\\
\label{Vprime}
&\mathbf V_{\mu^{(i)}}^\prime(t,b,\U )=1
-q_{pp}(q_{ps}^{(\A)}-q_{ps}^{(\I)})A_0(t+1).
\end{align}

\subsection*{Proof of $\eta_{\mu^{(i)}}(t,b,0)>\eta_{\mu^{(i)}}(t,b+1,0)$}
By substituting~(\ref{Gprime}) and~(\ref{Vprime}) in~(\ref{etaspecial}),
and noticing that
$\mathbf V_{\mu^{(i)}}^\prime(t,b,\U )=\mathbf V_{\mu^{(i)}}^\prime(t,b+1,\U )$
from~(\ref{Vprime})
and $\mathbf V_{\mu^{(i)}}^\prime(t,b,\U )>0$ (from Lemma~\ref{increasingWs} with $\mathbf D_{\mu}^\prime(\mathbf s)=0$),
the condition
$\eta_{\mu^{(i)}}(t,b,0)>\eta_{\mu^{(i)}}(t,b+1,0)$
is equivalent to
$\mathbf G_{\mu^{(i)}}^\prime(t,b,\U )>\mathbf G_{\mu^{(i)}}^\prime(t,b+1,\U )$,
which is readily verified from~(\ref{Gprime}), since
\begin{align}
&\mathbf G_{\mu^{(i)}}^\prime(t,b,\U )-\mathbf G_{\mu^{(i)}}^\prime(t,b+1,\U )
\\=&\nonumber
\left(q_{ps}^{(\A)}-q_{ps}^{(\I)}\right)R_{s\U}
\left[1-q_{pp}(1-q_{ps}^{(\I)})A_0(t+1)\right]
\nonumber
\\>&\nonumber
\left(q_{ps}^{(\A)}-q_{ps}^{(\I)}\right)
\frac{1-q_{pp}}{1-q_{pp}q_{ps}^{(\I)}}R_{s\U}
>0,
\end{align}
where the first inequality follows from the fact that
 $A_0(t+1)<\frac{1}{1-q_{pp}q_{ps}^{(\I)}}$,
the second from $q_{ps}^{(\I)}<q_{ps}^{(\A)}$.

\subsection*{Proof of $\eta_{\mu^{(i)}}(t,b,0)>\eta_{\mu}(t+1,b,0)$}
Since
$\mathbf V_{\mu^{(i)}}^\prime(t,b,\U)>0$, the condition 
$\eta_{\mu^{(i)}}(t,b,0)>\eta_{\mu}(t+1,b,0)$
is equivalent to
\begin{align}
& \mathbf G_{\mu^{(i)}}^\prime(t,b,\U )
\left(\mathbf V_{\mu^{(i)}}^\prime(t+1,b,\U )-\mathbf V_{\mu^{(i)}}^\prime(t,b,\U )\right)
\nonumber\\&
>
\mathbf V_{\mu^{(i)}}^\prime(t,b,\U )
\left(\mathbf G_{\mu^{(i)}}^\prime(t+1,b,\U )-\mathbf G_{\mu^{(i)}}^\prime(t,b,\U )\right).
\end{align}
Using~(\ref{Gprime}) and~(\ref{Vprime}), after algebraic manipulation we obtain
the equivalent condition
\begin{align}\label{condition}
&\left(1-q_{ps}^{(\A)}\right)p_{s,\mathrm{buf}}
+\left(1-q_{ps}^{(\A)}\right)
\left(q_{ps}^{(\A)}-q_{ps}^{(\I)}\right)b
\nonumber\\&
+\left(q_{ps}^{(\I)}-q_{ps}^{(\A)}\right)\Delta_s
>
0,
\end{align}
where we have used the fact that 
$T_{s\mathrm{K}}=\Delta_sR_{s\mathrm{U}}+T_{s\mathrm{U}}+p_{s,\mathrm{buf}}R_{s\mathrm{U}}$.
Since we require this condition to hold $\forall b\geq 0$
and the left hand expression is minimized by $b=0$,
the condition~(\ref{condition})
should be satisfied for $b=0$, yielding the equivalent condition
$\Delta_s
<
\frac{1-q_{ps}^{(\A)}}{q_{ps}^{(\A)}-q_{ps}^{(\I)}}p_{s,\mathrm{buf}}$,
which is an hypothesis of the lemma.

It is thus proved that the sequence of policies returned by Algorithm~\ref{algorithm}
 has the structure defined by~(\ref{struc}), where $b^{(i)}(t)$
satisfies the inequality~(\ref{incri}).
Moreover,
the inequality~(\ref{incrt}) holds since, by the algorithm construction, 
$\mu^{(i+1)}$ is obtained from $\mu^{(i)}$ by "activating" one additional state
from the set of idle states $\mathcal S_{\mathrm{idle}}^{(i)}$.

The second part of the lemma states that $b^{(N-1)}(t)=\bar b_{\max}(t)$,
 where $\bar b_{\max}(t)$ is given by
~(\ref{bmax}). This is a consequence of the fact that Algorithm~\ref{algorithm} stops if
the SU access efficiency becomes non-positive, \emph{i.e.},
 $\eta_{\mu^{(i)}}(\mathbf s)\leq 0$,
 $\forall \mathbf s\in\mathcal S_{\mathrm{idle}}^{(i)}$.
From~(\ref{etaspecial}),
this condition is equivalent to
$\mathbf G_{\mu^{(i)}}^\prime(t,b,\U)\leq 0$,
$\forall (t,b,\U)\in\mathcal S_{\mathrm{idle}}^{(i)}$.
By using~(\ref{Gprime})
and by solving $\mathbf G_{\mu^{(i)}}^\prime(t,b,\U)\leq 0$
with respect to $b$, the result follows.
\end{proof}

\bibliographystyle{IEEEtran}
\bibliography{IEEEabrv,biblio}
\begin{IEEEbiography}[{\includegraphics[width=1in,height=1.25in,clip,keepaspectratio]{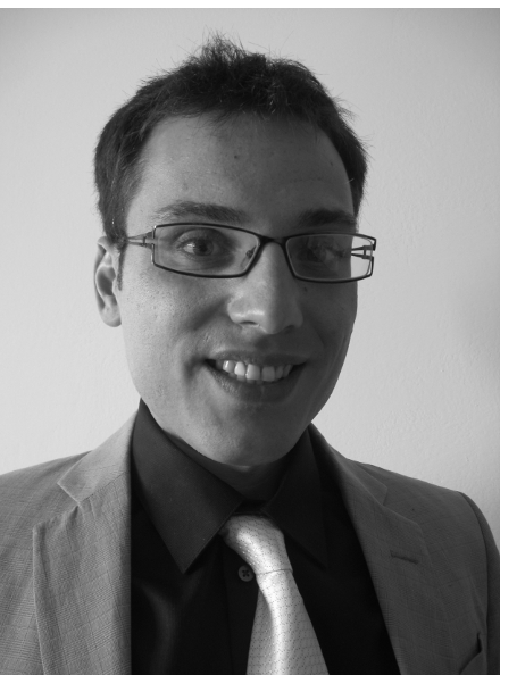}}]
{Nicol\`{o}~Michelusi} (S'09) received the B.S. (Electronics Engineering) and M.S.
 (Telecommunications Engineering) degrees summa cum laude from the University of Padova, Italy, in 2006 and 2009, respectively,
 and the M.S. degree in Telecommunications Engineering from Technical University of Denmark, Copenhagen, Denmark, in 2009, as part of the T.I.M.E. double degree program.
Since January 2009, he is a Ph.D. student at University of Padova, Italy.
In 2011, he was on leave at the University
of Southern California, Los Angeles, United States, as a visiting Ph.D. student.
 His research interests include ultrawideband communications, wireless networks, cognitive radio networks, stochastic optimization, optimal control,
 energy harvesting for communications.
\end{IEEEbiography}
\begin{IEEEbiography}[{\includegraphics[width=1in,height=1.25in,clip,keepaspectratio]{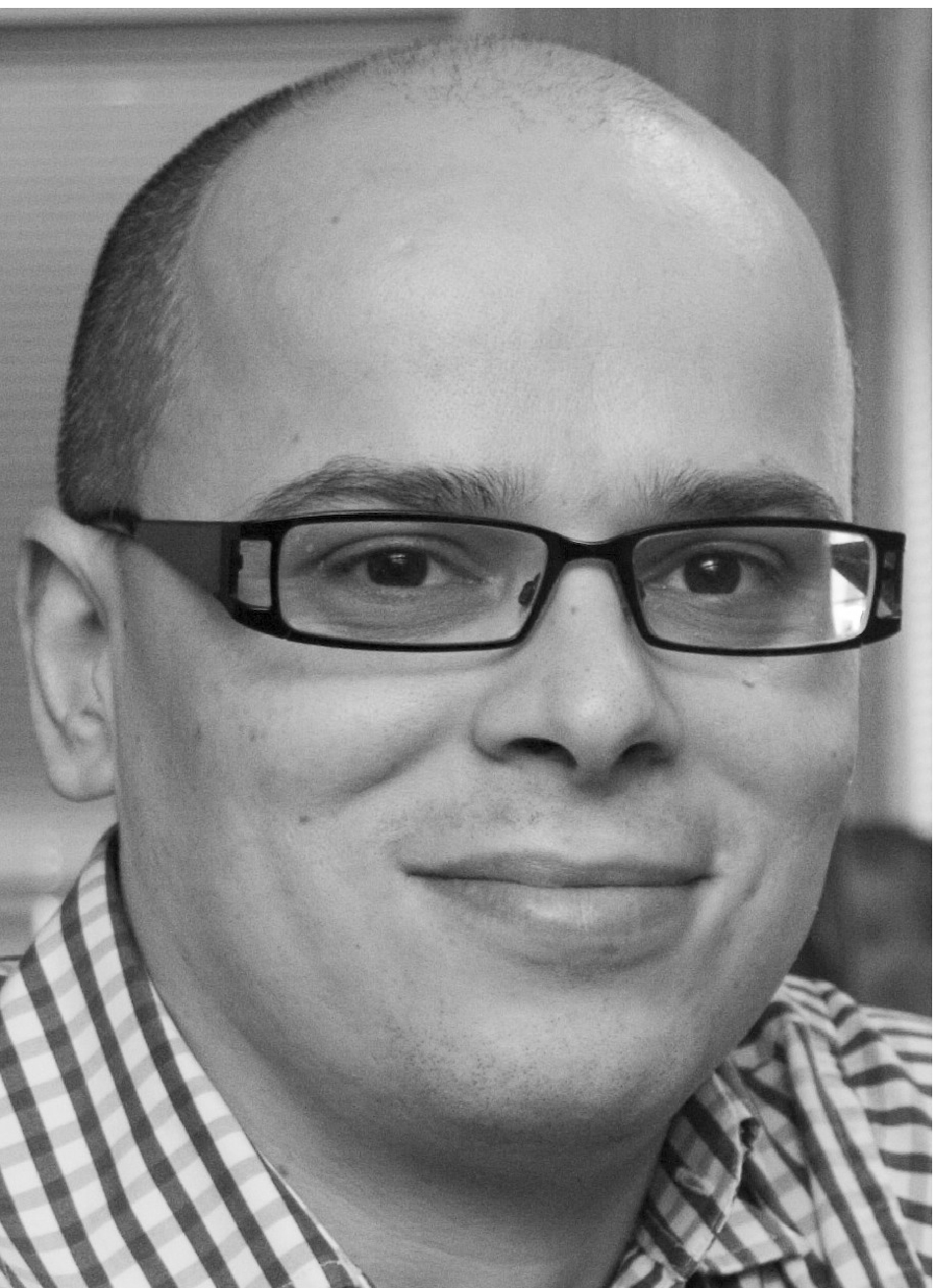}}]
{Petar Popovski} (S'97-A'98-M'04-SM'10) received the Dipl.-Ing. in
electrical engineering and Magister Ing. in communication engineering
from Sts. Cyril and Methodius University, Skopje, Macedonia, in 1997
and 2000, respectively and Ph. D. from Aalborg University, Denmark, in
2004. He was Assistant Professor (2004-2009) and Associate Professor
(2009-2012) at Aalborg University. From 2008 to 2009 he held part-time
position as a wireless architect at Oticon A/S. Since 2012 he is a
Professor at Aalborg University. He has more than 140 publications in
journals, conference proceedings and books and has more than 25
patents and patent applications. He has received the Young Elite
Researcher award and the SAPERE AUDE career grant from the Danish
Council for Independent Research. He has received six best paper
awards, including three from IEEE. Dr. Popovski serves on the
editorial board of several journals, including IEEE Communications
Letters (Senior Editor), IEEE Transactions on Communications and IEEE
Transactions on Wireless Communications.  His research interests are
in the broad area of wireless communication and networking,
information theory and protocol design.
\end{IEEEbiography}
\begin{IEEEbiography}[{\includegraphics[width=1in,height=1.25in,clip,keepaspectratio]{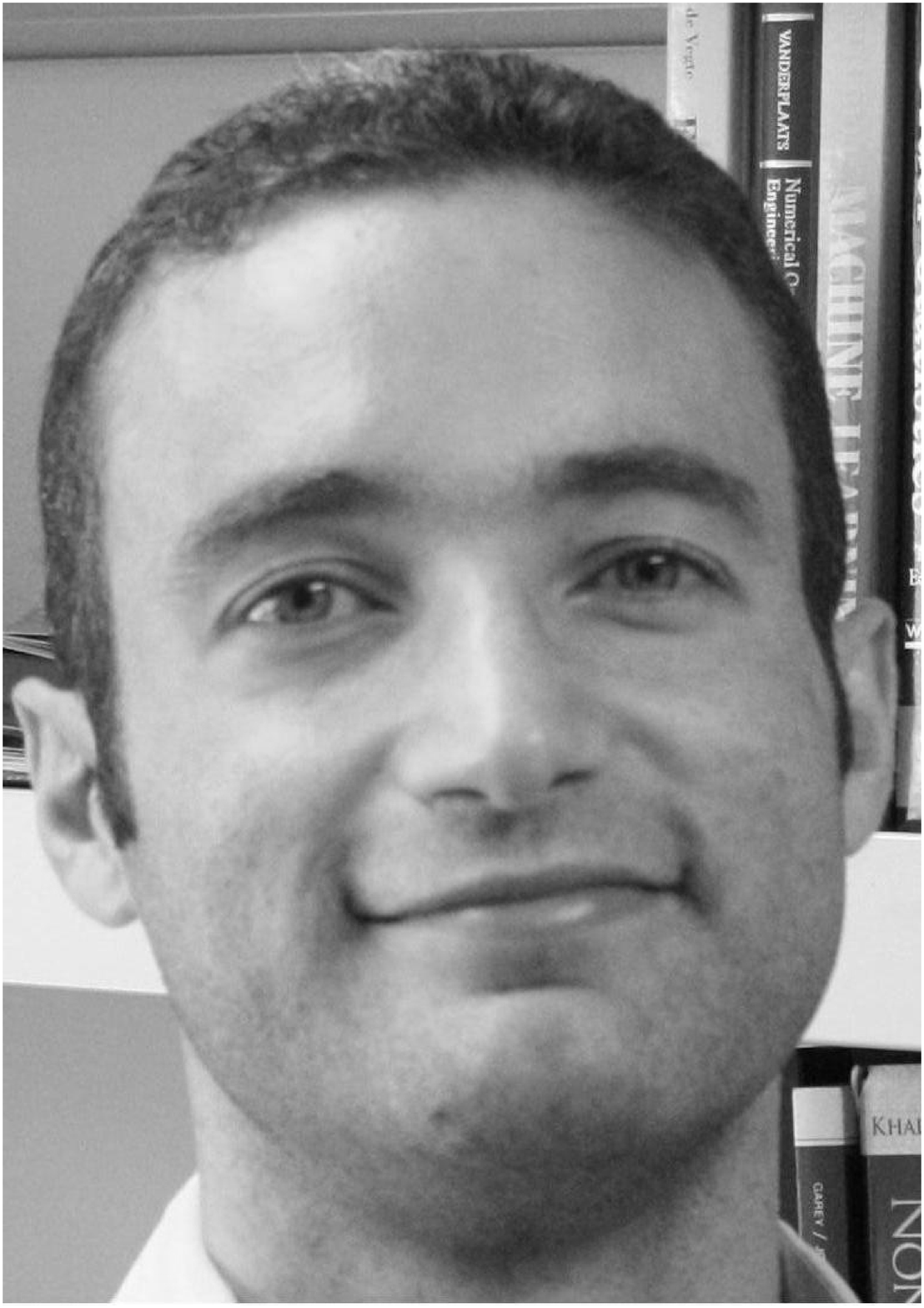}}]{Osvaldo Simeone} (M'02)
 received the M.Sc. degree (with honors) and the Ph.D. degree in information engineering from Politecnico di Milano, Milan,
Italy, in 2001 and 2005, respectively. He is currently with the Center for Wireless Communications and Signal Processing Research (CWCSPR),
New Jersey Institute of Technology (NJIT), Newark, where he is an Associate Professor. His current research interests concern the cross-layer analysis and design of wireless networks with emphasis on information-theoretic, signal processing, and queuing aspects.
Specific topics of interest are: cognitive radio, cooperative communications, rate-distortion theory, ad hoc, sensor, mesh and
hybrid networks, distributed estimation, and synchronization. Dr. Simeone is a co-recipient of Best Paper Awards of the IEEE SPAWC 2007
and IEEE WRECOM 2007. He currently serves as an Editor for IEEE TRANSACTIONS ON COMMUNICATIONS.
\end{IEEEbiography}
\begin{IEEEbiography}[{\includegraphics[width=1in,height=1.25in,clip,keepaspectratio]{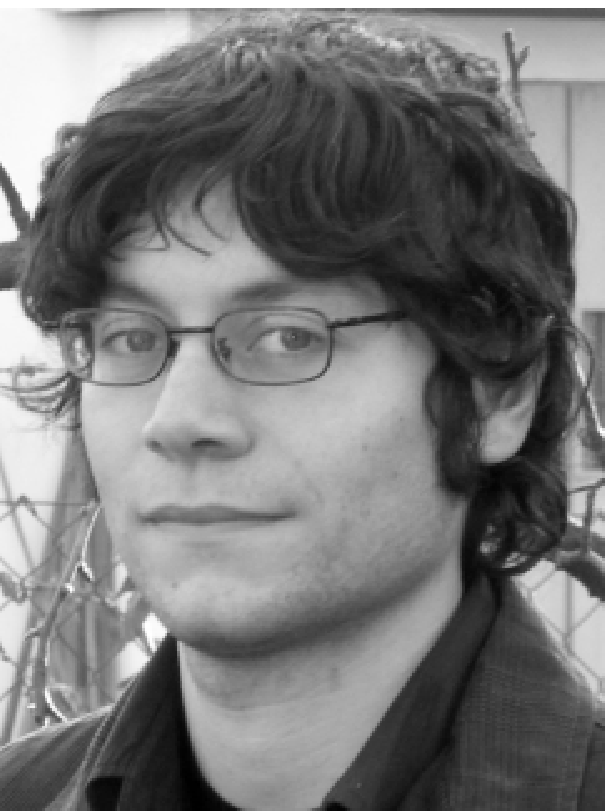}}]
{Marco Levorato} (S'06, M'09) obtained both the
BE (Electronics and Telecommunications Engineer-
ing) and the ME (Telecommunications Engineering)
summa cum laude from the University of Ferrara,
Italy, in 2002 and 2005, respectively. In 2009, he
received a Ph.D. in Information Engineering from
the University of Padova. During 2008 he was on
leave at the University of Southern California, Los
Angeles, United States. In 2009 he was a post
doctorate researcher at the University of Padova.
Since January 2010, he is a post doctorate researcher
at Stanford and the University of Southern California (USC).
\end{IEEEbiography}
\begin{IEEEbiography}[{\includegraphics[width=1in,height=1.25in,clip,keepaspectratio]{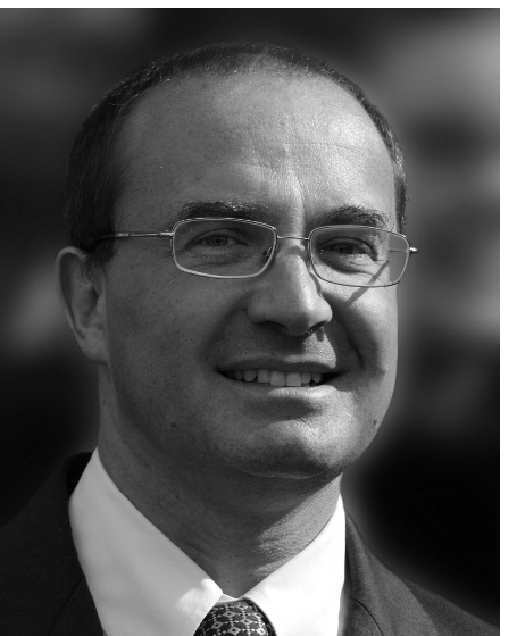}}]
{Michele Zorzi} (S'89, M'95, SM'98, F'07) was born in Venice, Italy, on December 6th, 1966. He received the Laurea and
 the PhD degrees in Electrical Engineering from the University of Padova, Italy, in 1990 and 1994, respectively.
 During the Academic Year 1992/93, he was on leave at the University of California, San Diego (UCSD) as a visiting PhD student,
 working on multiple access in mobile radio networks. In 1993, he joined the faculty of the Dipartimento di Elettronica e 
Informazione, Politecnico di Milano, Italy. After spending three years with the Center for Wireless Communications at UCSD,
 in 1998 he joined the School of Engineering of the University of Ferrara, Italy, where he became a Professor in 2000.
 Since November 2003, he has been on the faculty at the Information Engineering Department of the University of Padova.
 His present research interests include performance evaluation in mobile communications systems, random access in mobile radio
 networks, ad hoc and sensor networks, energy constrained communications protocols, broadband wireless access and underwater 
acoustic communications and networking.

Dr. Zorzi was the Editor-In-Chief of the IEEE WIRELESS COMMUNICATIONS MAGAZINE from 2003 to 2005 and the Editor-In-Chief of the
 IEEE TRANSACTIONS ON COMMUNICATIONS from 2008 to 2011, and currently serves on the Editorial Board of the WILEY JOURNAL OF
 WIRELESS COMMUNICATIONS AND MOBILE COMPUTING. He was also guest editor for special issues in the IEEE PERSONAL COMMUNICATIONS 
MAGAZINE (Energy Management in Personal Communications Systems) IEEE WIRELESS COMMUNICATIONS MAGAZINE (Cognitive Wireless
 Networks) and the IEEE JOURNAL ON SELECTED AREAS IN COMMUNICATIONS (Multi-media Network Radios, and Underwater Wireless 
Communications Networks). He served as a Member-at-large of the Board of Governors of the IEEE Communications Society 
from 2009 to 2011.
\end{IEEEbiography}
\end{document}